%% file: main.tex
\pdfoutput=1
\documentclass[sigconf,screen,nonacm]{acmart}

\renewcommand\footnotetextcopyrightpermission[1]{}

\usepackage{booktabs}
\usepackage{amsmath}
\usepackage{mathtools}
\usepackage{algorithm}
\usepackage{algpseudocode}
\usepackage{tikz}
\usepackage{pgfplots}
\usetikzlibrary{arrows.meta,backgrounds,calc,decorations.pathreplacing,fit,positioning}
\pgfplotsset{compat=1.17}

\definecolor{FDdark}{RGB}{27,67,150}
\definecolor{FDred}{RGB}{234,50,35}
\definecolor{FDgray}{RGB}{98,104,109}
\definecolor{FDpale}{RGB}{242,243,244}

\newcommand{\hprcrst}{\textsc{HP-RCRST}}
\newcommand{\treecost}{L}
\newcommand{\treeradius}{R}
\newcommand{\mstcost}{M}
\newcommand{\directradius}{\Delta}

\title[Provably Good Prim--Dijkstra Revisited]{
\texorpdfstring{
Provably Good Prim--Dijkstra Revisited:\\
New Theory and a Practical Algorithm for a\\
Classical VLSI Routing Problem with LLMs
}{Provably Good Prim--Dijkstra Revisited: New Theory and a Practical Algorithm
for a Classical VLSI Routing Problem with LLMs}}

\author{Keren Zhu}
\email{krzhu@fudan.edu.cn}
\affiliation{
  \department{College of Integrated Circuits and Nano-Micro Electronics}
  \department{State Key Laboratory of Integrated Chips and System}
  \institution{Fudan University}
  \city{Shanghai}
  \country{China}
}

\begin{document}

\begin{abstract}
Large language models may make precise but dormant algorithmic problems
practical to revisit, and may expose new paths toward fundamental ones.  We
demonstrate this possibility through Prim--Dijkstra routing, a classic VLSI
problem whose terminal-only Manhattan complexity remained open despite
decades of practical work.  We prove weak NP-completeness, derive a continuous
cost--radius tradeoff with a balanced \((2,2)\) guarantee, and build
\hprcrst, a height-partition-based multi-mode solver.  On 28 development
instances, its stronger
modes Pareto-dominate the published-method union on 23 and tie on five.  The
case shows how conflicting conjectures, counterexamples, formal checks, and
implementation can reopen neglected questions.  Code and reproducibility
materials are available at \url{https://github.com/CODA-Team/hp-rcrst}.
\end{abstract}


\maketitle

\input{introduction}
\input{classical_problem}
\input{research_process}
\input{theory}
\input{algorithm}
\input{experiments}
\input{discussion}
\input{conclusion}

\bibliographystyle{ACM-Reference-Format}
\bibliography{references}

\clearpage
\appendix
\input{appendix_npc}
\input{appendix_height}
\input{appendix_algorithm}
\input{appendix_experiments}
\clearpage
\input{appendix_record}

\end{document}

%% file: introduction.tex
\section{Introduction}
\label{sec:introduction}

Some questions wait in the literature.  Others are put to work while they
wait.  Prim--Dijkstra routing has lived in both worlds for more than three
decades.

A routing tree may be short when viewed from the plane and long when viewed
from its source.  A minimum spanning tree minimizes total edge length but can
send a source--sink path through a long detour; a shortest-path tree minimizes
those paths but may pay for many direct connections.  Prim and Dijkstra occupy
the two ends of this tension.  Both are among the first greedy algorithms a
student encounters.  Putting them on the same dial creates a problem that can
be stated in a few lines and has remained part of physical-design research
ever since.  Figure~\ref{fig:mst-spt} shows the disagreement on a
small instance and the kind of tree sought between the two endpoints.

\begin{figure*}[t]
  \centering
  \begin{tikzpicture}[
      x=0.72cm,y=0.62cm,
      tree/.style={line width=0.95pt,draw=black!72},
      shortcut/.style={line width=1.55pt,draw=FDdark},
      terminal/.style={circle,draw=black!70,fill=white,
                       inner sep=0pt,minimum size=3.8pt},
      root/.style={rectangle,draw=black,fill=black,
                   inner sep=0pt,minimum size=5pt},
      target/.style={circle,draw=FDdark!78!black,fill=FDdark!14,
                     inner sep=0pt,minimum size=5pt},
      panel/.style={font=\small\bfseries},
      metric/.style={font=\small}
    ]
    \begin{scope}
      \draw[black!8,step=1] (-0.25,-2.25) grid (4.25,2.25);
      \draw[tree]
        (0,0)--(0,1)--(0,2)--(1,2)--(2,2)--(2,1)--(2,0)
        --(3,0)--(4,0)--(4,-1)--(4,-2)--(3,-2)--(2,-2);
      \foreach \p in {(0,1),(0,2),(1,2),(2,2),(2,1),(2,0),
                      (3,0),(4,0),(4,-1),(4,-2),(3,-2)}
        \node[terminal] at \p {};
      \node[root,label={[font=\footnotesize]below left:\(r\)}] at (0,0) {};
      \node[target,label={[font=\footnotesize]below:\(t\)}] at (2,-2) {};
      \node[panel] at (2,2.82) {Minimum spanning tree};
      \node[metric] at (2,-2.85)
        {\(\treecost=12,\quad \treeradius=12\)};
    \end{scope}

    \begin{scope}[shift={(7.3,0)}]
      \draw[black!8,step=1] (-0.25,-2.25) grid (4.25,2.25);
      \draw[tree] (0,0)--(0,1)--(0,2)--(1,2)--(2,2);
      \draw[tree]
        (2,1)--(2,0)--(3,0)--(4,0)--(4,-1)--(4,-2)--(3,-2)--(2,-2);
      \draw[shortcut] (0,0)--node[midway,above,font=\footnotesize,
        text=FDdark!82!black] {\(2\)} (2,0);
      \foreach \p in {(0,1),(0,2),(1,2),(2,2),(2,1),(2,0),
                      (3,0),(4,0),(4,-1),(4,-2),(3,-2)}
        \node[terminal] at \p {};
      \node[root,label={[font=\footnotesize]below left:\(r\)}] at (0,0) {};
      \node[target,label={[font=\footnotesize]below:\(t\)}] at (2,-2) {};
      \node[panel] at (2,2.82) {A cost--radius compromise};
      \node[metric] at (2,-2.85)
        {\(\treecost=13,\quad \treeradius=8\)};
    \end{scope}

    \begin{scope}[shift={(14.6,0)}]
      \draw[black!8,step=1] (-0.25,-2.25) grid (4.25,2.25);
      \draw[tree] (0,0)--(0,1)--(0,2)--(1,2)--(2,2);
      \draw[tree] (2,1)--(2,0)--(3,0)--(4,0)--(4,-1);
      \draw[tree] (2,-2)--(3,-2)--(4,-2);
      \draw[shortcut] (0,0)--node[midway,above,font=\footnotesize,
        text=FDdark!82!black] {\(2\)} (2,0);
      \draw[shortcut,rounded corners=1.5pt]
        (0,0)--node[midway,left,font=\footnotesize,
        text=FDdark!82!black] {\(4\)} (0,-2)--(2,-2);
      \foreach \p in {(0,1),(0,2),(1,2),(2,2),(2,1),(2,0),
                      (3,0),(4,0),(4,-1),(4,-2),(3,-2)}
        \node[terminal] at \p {};
      \node[root,label={[font=\footnotesize]below left:\(r\)}] at (0,0) {};
      \node[target,label={[font=\footnotesize]below:\(t\)}] at (2,-2) {};
      \node[panel] at (2,2.82) {Shortest-path tree};
      \node[metric] at (2,-2.85)
        {\(\treecost=16,\quad \treeradius=6=\directradius\)};
    \end{scope}
  \end{tikzpicture}
  \caption{The Prim--Dijkstra problem between its two textbook endpoints.
  The same 13 terminals are used in all three panels.  Unit chain edges form
  an MST with \((\treecost,\treeradius)=(12,12)\).  One Manhattan shortcut
  gives the middle tree \((13,8)\); a second produces an SPT \((16,6)\).
  Labeled heavy paths depict complete-metric edges with the shown Manhattan
  lengths; an unlabeled bend is not a Steiner vertex.  The cost--radius
  problem asks which middle-ground trees satisfy prescribed bounds on both
  quantities.}
  \Description{Three trees on the same double-hairpin arrangement of 13
  terminals.  The minimum spanning tree follows the full chain.  A compromise
  tree adds one shortcut from the root.  The shortest-path tree uses two root
  shortcuts, increasing total length while reducing source radius.}
  \label{fig:mst-spt}
\end{figure*}
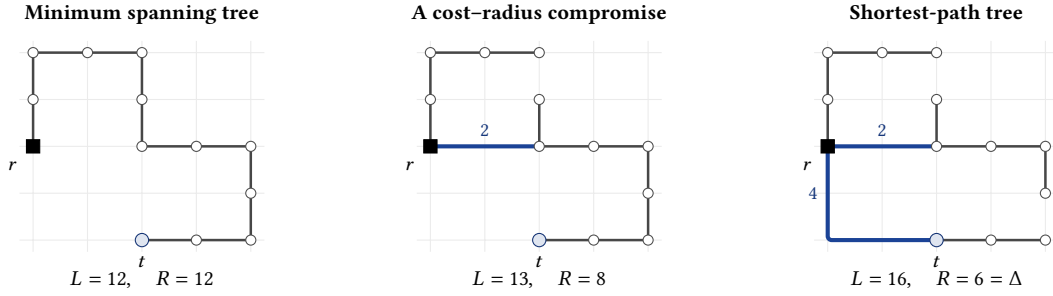

The early performance-driven routing literature did not treat this tension as
mere engineering folklore.  In 1992, Cong et al.\ formulated bounded-radius
routing trees and developed algorithms that were, in the language of their
title, \emph{provably good}~\cite{CongEtAl1992ProvablyGood}.  Their paper also
left the complexity of the minimum-cost bounded-radius spanning tree in the
Manhattan plane as an open question.  Three years later, Alpert et al.\ joined
Prim and Dijkstra in a single parameterized construction~\cite{AlpertEtAl1995PrimDijkstra}.
The resulting PD family became a durable way to navigate the wirelength--delay
tradeoff.

The line did not end.  Shallow-light trees gave a general-graph guarantee by
changing the form of the question~\cite{KhullerEtAl1995Balancing}.  Allowing
Steiner vertices changes the attainable lightness, and geometric variants
continue to refine that distinction
~\cite{ElkinSolomon2015SteinerShallowLight,Solomon2015EuclideanSteiner,LeEtAl2026EuclideanSLT}.
SALT carried provable shallow-light reasoning into rectilinear Steiner
routing~\cite{ChenYoung2020SALT}.
In 2018, \emph{Prim--Dijkstra Revisited} showed that the elementary PD
construction still admitted substantial practical improvement~\cite{AlpertEtAl2018PDII}.
Recent multi-source variants moved the objective toward cost--skew
tradeoffs~\cite{KahngThumathyWoo2023}.  Across these changes, the same two
quantities kept returning: how much tree is built, and how far the tree makes
a signal travel.

Yet the original terminal-only Manhattan problem remained between two
communities.  For EDA it was familiar enough that a heuristic could be used
without first settling its exact complexity.  For theoretical computer
science it looked application-specific beside the general theory of
shallow-light trees.  The result was an unusual asymmetry: the construction
was repeatedly improved, but the most basic complexity question and a direct
global-radius approximation theory were not brought together with a modern
implementation.

No contemporary routing flow is waiting for this classification.  Its value
now lies in the class of problems it represents: precise, locally checkable
questions that slipped out of active view as communities and tools moved on.
Reopening such a question once required enough literature recovery, proof
work, and software reconstruction to make the attempt difficult to justify.
LLMs change that threshold.  More of these questions can now be revisited,
and some may reveal theoretical or practical opportunities that were never
visible when they were set aside.  The same parallel cycle of conjecture,
attack, formalization, and experiment may also open new routes into
fundamental problems.

Recent language-model research now touches two neighboring forms of inquiry.
In mathematics, OpenAI has released a GPT-generated proof of the cycle double
cover conjecture together with the prompt that produced
it~\cite{OpenAI2026CycleDoubleCover}, while formal-mathematics agents are
beginning to address open-ended research questions
~\cite{TaoEtAl2026SolversToResearch}.  AlphaGeometry couples learned guidance
with symbolic deduction, while LeanDojo couples language models to a
proof-assistant environment
~\cite{TrinhEtAl2024AlphaGeometry,YangEtAl2023LeanDojo}.
Evaluator-guided systems such as FunSearch and AlphaEvolve treat programs and
algorithms as objects of discovery
~\cite{RomeraParedesEtAl2024FunSearch,DeepMind2025AlphaEvolve}.  In EDA,
multi-agent self-evolution modifies, evaluates, and retains improvements to the
design tools themselves~\cite{YuRen2026SelfEvolvedABC}.

Mathematics and EDA approach the new machinery from opposite directions: one
through conjectures and proofs, the other through executable tools and
objective evaluators.  The old Prim--Dijkstra question sits between them.  It
is a theorem question embedded in a practical routing construction.  The
construction has been repeatedly used and improved; the complexity question
remains where it was left.

That position makes the problem unusually suitable for a model-assisted
revisit.  Its definition is short; candidate trees and counterexamples are
inexpensive to check; mathematical claims can be tested against executable
instances; and the implementation has a small, inspectable engineering
surface.  Proof validity, problem choice, and final research judgment remain
the author's responsibility.

For this revisit, we organized a set of deliberately conflicting investigations.
Some were asked to prove polynomial-time solvability without access to prior
answers; another was asked to establish hardness; others pursued approximation
and implementation with progressively broader information.  Failed claims
were preserved as counterexamples, and the empirical objective was frozen
before the final solver was developed.  The author selected the problem and
experimental conditions, reconstructed the arguments, performed the literature
due diligence, and verified the claims reported here.

Our contributions are as follows:

\begin{itemize}
  \item We settle the complexity question with an explicit integer reduction.
  It proves NP-completeness for the rooted, terminal-only Manhattan
  cost--radius problem.  The reduction is weak and numerical.  To the best of
  our knowledge, this answers the question left open in 1992.
  \item We give a bottom-up height partition of a rooted MST.  For a continuous
  parameter \(H\), it bounds global radius by \(\directradius+H\) and excess
  length through \(H\); the balanced choice gives a \((2,2)\) cost--radius
  guarantee.  The parameter \(H\) continues the sequence of knobs that has shaped
  this literature.
  \item We implement the theory in \hprcrst, a self-contained deterministic
  solver with one certified mode and three practical runtime--quality modes.
  On the 28 development instances, its stronger modes Pareto-dominate the
  published-method union on 23 cases and tie it on the other five; the full
  mode-by-mode comparison appears in Section~\ref{sec:experiments}.
  \item We release the five task specifications and information boundaries
  used in the study, together with the final solver and checker, selected
  benchmarks, baseline adaptations, frozen result bundles, and the completed
  sorry-free Lean modules.  The public artifact is a reproducibility subset
  of the research record, not a transcript of every model interaction.
\end{itemize}

%% file: classical_problem.tex
\section{A Classical Cost--Radius Problem}
\label{sec:classical}

\subsection{Model}

Let \(V\subset\mathbb{Z}^2\) be a set of distinct terminals and let
\(r\in V\) be the source.  Edges belong to the complete terminal graph and
have Manhattan length \(d(u,v)\).  For a spanning tree \(T\), let
\[
  \treecost(T)=\sum_{e\in T}d(e)
  \qquad\text{and}\qquad
  \treeradius(T)=\max_{v\in V}d_T(r,v).
\]
The decision problem asks whether a tree exists with
\(\treecost(T)\le B\) and \(\treeradius(T)\le D\).  We write
\(\mstcost\) for the exact terminal MST length and
\(\directradius=\max_v d(r,v)\).  Every feasible comparison tree satisfies
\(\treecost(T)\ge\mstcost\) and
\(\treeradius(T)\ge\directradius\).
We therefore summarize both objectives by the normalized performance vector
\[
  \boldsymbol{\rho}(T)
  :=
  \left(
    \frac{\treecost(T)}{\mstcost},
    \frac{\treeradius(T)}{\directradius}
  \right),
  \qquad
  \boldsymbol{\rho}(T)\succeq(1,1),
\]
where vector inequalities are componentwise and smaller is better.

\begin{example}[The running instance]
\label{ex:running}
For the 13 terminals in Figure~\ref{fig:mst-spt},
\(\mstcost=12\) and \(\directradius=6\).  The three displayed trees have
\[
  \boldsymbol{\rho}(T_{\rm MST})=(1,2),\qquad
  \boldsymbol{\rho}(T_{\rm mid})=
     \left(\frac{13}{12},\frac43\right),\qquad
  \boldsymbol{\rho}(T_{\rm SPT})=
     \left(\frac43,1\right).
\]
With budgets \(B=13\) and \(D=8\), the MST violates only the radius bound,
the SPT violates only the length bound, and the middle tree satisfies both.
Thus neither endpoint solves even this small rectangular decision instance.
\end{example}

This paper keeps three boundaries explicit.  Trees are terminal-only; the
radius is measured from a designated source rather than by all-pairs diameter;
and the same Manhattan metric determines both total cost and path length.
General weighted graphs, Euclidean geometry, Steiner points, per-terminal
stretch, and skew objectives are important neighbors, but they are not the
same problem.

\subsection{Constructions, knobs, and guarantees}

The two endpoints optimize one coordinate each:
\[
  L(T_{\rm MST})=M,
  \qquad
  R(T_{\rm SPT})=\Delta .
\]
\paragraph{Metric status.}
Every empirical point and every tree passed to our evaluator uses the complete
terminal graph with \(d(u,v)=\lVert u-v\rVert_1\).  The endpoint identities
and our height-partition theorem hold in any finite metric.  KRY/LAST is
stated more generally for a nonnegatively weighted connected graph, and hence
applies to this complete Manhattan graph.  BRBC, PD, and PD-II were developed
in the rectilinear VLSI-routing setting, while MSPD/MSS also changes the source
and skew objective.  When these methods enter our comparison, we retain only
their terminal parent trees and recompute \(L\) and \(R\) in the present
single-source Manhattan model.  A common metric therefore does not mean that
the original papers solve the same optimization problem.

KRY recasts the tradeoff as a light approximate shortest-path tree
(LAST)~\cite{KhullerEtAl1995Balancing}.  For every \(\alpha>1\), it constructs
a tree \(T_\alpha\) satisfying
\[
  d_{T_\alpha}(r,v)\leq\alpha\,d(r,v)\quad(\forall v),
  \qquad
  L(T_\alpha)\leq
  \left(1+\frac{2}{\alpha-1}\right)M.
\]
Consequently,
\[
  \boldsymbol{\rho}(T_\alpha)
  \preceq
  \left(1+\frac{2}{\alpha-1},\,\alpha\right).
\]
The symmetric setting is
\(\alpha=1+\sqrt2\), giving the common factor \(1+\sqrt2\).
This is stronger than a global-radius statement on its path coordinate: it
controls every terminal relative to that terminal's own shortest distance.

The EDA constructions use their knobs differently.  The bounded-radius,
bounded-cost (BRBC) family starts from an MST traversal and inserts root
connections when a radius threshold is crossed
~\cite{CongEtAl1992ProvablyGood}.  Prim--Dijkstra (PD) instead grows one tree
with the key
\[
  \operatorname{key}_{\lambda}(u,v)
    = d_T(r,u)+\lambda\,d(u,v),\qquad \lambda\geq1,
\]
moving from Dijkstra-like path growth toward Prim-like short edges as
\(\lambda\) increases~\cite{AlpertEtAl1995PrimDijkstra}.  PD-II uses separate
parameters for construction and detour repair, then applies edge flipping
~\cite{AlpertEtAl2018PDII}; MSPD/MSS further varies the source set and targets
cost--skew rather than the exact rectangle studied here
~\cite{KahngThumathyWoo2023}.  These parameter values are not interchangeable,
and PD and PD-II do not inherit the LAST bound above.

Our parameter \(H\) returns to a direct global cost--radius statement:
\[
  \boldsymbol{\rho}(A_H)
  \preceq
  \left(1+\frac{\Delta}{H},\,1+\frac{H}{\Delta}\right),
  \qquad
  H=\Delta\ \Longrightarrow\ \boldsymbol{\rho}(A_H)\preceq(2,2).
\]
Thus KRY/LAST and \hprcrst\ trace the envelopes
\((x-1)(y-1)=2\) and \((x-1)(y-1)=1\), respectively, in the normalized
cost--radius plane.  The second curve is tighter for global radius; it does
not replace KRY's stronger per-terminal stretch theorem.

\begin{figure}[t]
  \centering
  \begin{tikzpicture}
    \begin{axis}[
        width=\columnwidth,height=5.7cm,
        xmin=0.95,xmax=3.65,ymin=0.95,ymax=3.65,
        xlabel={normalized cost \(L/M\)},
        ylabel={normalized radius \(R/\Delta\)},
        xtick={1,2,3},ytick={1,2,3},
        tick label style={font=\footnotesize},
        label style={font=\small},
        axis line style={black!45},
        grid=both,grid style={black!7},
        legend style={
          at={(0.5,-0.29)},anchor=north,
          legend columns=2,draw=none,
          font=\footnotesize,column sep=7pt,
          cells={anchor=west}
        },
        clip=false
      ]
      \addplot[black!50,dashed,line width=1.0pt,domain=1.76:3.6,
               samples=90] {1+2/(x-1)};
      \addlegendentry{KRY/LAST}

      \addplot[FDdark,line width=1.4pt,domain=1.38:3.6,
               samples=90] {1+1/(x-1)};
      \addlegendentry{\hprcrst\ (\(H\))}
      \addplot[forget plot,only marks,mark=square*,mark size=2.3pt,
               mark options={fill=FDdark,draw=FDdark!90!black}]
        coordinates {(3,1.5) (2,2) (1.5,3)};
      \node[font=\footnotesize,anchor=south west,text=FDdark!90!black,
            fill=white,fill opacity=0.86,text opacity=1,inner sep=1pt]
        at (axis cs:2.06,2.08) {\((2,2)\)};

      \addplot[black!58,densely dotted,line width=0.9pt,
               mark=o,mark size=2.0pt,
               mark options={fill=white,draw=black!70}]
        coordinates {
          (1.18,3.32) (1.42,2.88) (1.72,2.47)
          (2.08,2.06) (2.52,1.63) (3.25,1.16)
        };
      \addlegendentry{PD (\(\lambda\))}

      \draw[-{Stealth[length=3.2pt]},FDred!78!black,line width=0.9pt]
        (axis cs:1.42,2.88)--(axis cs:1.57,2.64);
      \draw[-{Stealth[length=3.2pt]},FDred!78!black,line width=0.9pt]
        (axis cs:2.08,2.06)--(axis cs:1.90,2.22);
      \draw[-{Stealth[length=3.2pt]},FDred!78!black,line width=0.9pt]
        (axis cs:2.52,1.63)--(axis cs:2.72,1.51);
      \addplot[only marks,mark=diamond*,mark size=2.2pt,
               mark options={fill=FDred!70,draw=FDred!90!black}]
        coordinates {(1.57,2.64) (1.90,2.22) (2.72,1.51)};
      \addlegendentry{PD-II (\(\alpha_1,\alpha_2\))}

      \node[font=\footnotesize,anchor=north east,text=black!50]
        at (axis cs:3.55,3.55) {\(\swarrow\) better};
    \end{axis}
  \end{tikzpicture}
  \caption{The knob lineage on one normalized cost--radius plane.  The
  KRY/LAST and \hprcrst\ curves are worst-case upper-bound envelopes from
  the displayed bounds, not empirical frontiers.  Hollow points and red repair arrows
  show only the roles of the PD and PD-II parameters; their plotted positions
  are schematic and carry neither a numerical claim nor a method ranking.}
  \Description{A normalized cost-radius plot. A dashed KRY curve lies above
  a blue HP-RCRST curve through the point 2,2. Schematic hollow PD points run
  between the MST and shortest-path ends, and red arrows show PD-II repairs
  toward lower cost and radius.}
  \label{fig:knob-lineage}
\end{figure}
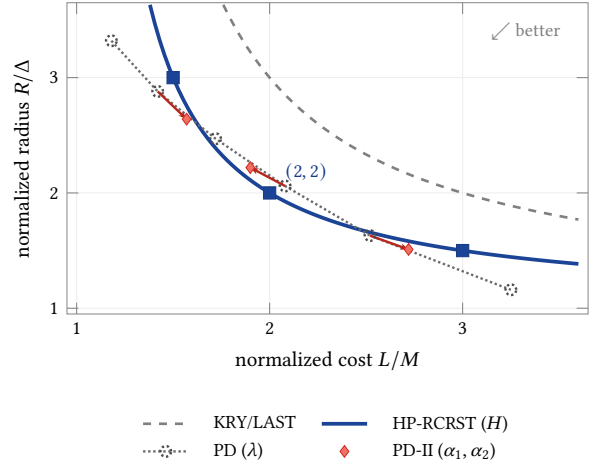

\subsection{Neighboring models}

The related theory follows several parallel routes.  In general graphs,
delay-constrained spanning trees were shown hard and studied through network
heuristics, mathematical programming, and bi-objective optimization
~\cite{SalamaEtAl1997DelayConstrained,GouveiaEtAl2008RootedDistance,CarvalhoCoco2023Biobjective}.
Those formulations permit an arbitrary input graph and, in several cases,
separate cost and delay weights.  They therefore do not settle the complete
terminal graph in which one Manhattan metric supplies both quantities.

Geometric variants also impose different path constraints.  Bounded-path
spanning and Steiner trees have been treated by exact and heuristic methods
~\cite{OhEtAl1997BoundedPath}, and single-source dilation bounds compare each
root path with that terminal's direct Euclidean distance
~\cite{CheongLee2013SingleSource}.  Diameter and minimum-radius variants supply
further hardness results
~\cite{HoEtAl1991MinimumDiameter,SeoEtAl2009Geometric}.  A per-terminal
multiplicative dilation bound, an all-pairs diameter bound, and our single
global additive budget \(D\) define distinct feasible sets.

A second neighboring line replaces the pair of hard bounds by a scalar
cost--distance objective: total construction cost plus weighted source--sink
path lengths.  It begins with two-metric network design
~\cite{MeyersonEtAl2008CostDistance}; uniform variants admit progressively
sharper approximations
~\cite{KhazraeiHeld2021UniformCostDistance,FoosEtAl2023UniformCostDistance}.
Related EDA work embeds timing, congestion, buffering, or multi-net resource
sharing into larger optimization systems
~\cite{HeldEtAl2018TimingRouting,DaboulEtAl2023GlobalInterconnect,HeldPerner2025CostDistance}.
These models explain the practical reach of the same cost--delay tension, but
their weighted-sum objectives, Steiner vertices, routing graphs, and global
design constraints are not the terminal-only Manhattan decision problem
studied here.  Thus the results above narrow, rather than remove, the novelty
of our theorem: its role is to close that precise case and connect the answer
to a particularly simple approximation.

%% file: research_process.tex
\section{Revisiting the Problem with Language Models}
\label{sec:process}

We used the same model---OpenAI GPT-5.6 Sol in Codex with ultra reasoning
effort---under five deliberately incompatible research conditions.  What
changed was the premise, the information boundary, and the object that had to
survive external checking:

\begin{description}
  \item[Hardness.] Starting from the problem definition alone, construct and
  adversarially audit an NP-completeness proof.
  \item[Exact, blind.] Pursue a polynomial-time exact algorithm without
  literature, prior algorithms, Internet access, or permission to retreat to
  a hardness claim.
  \item[Exact, informed.] Pursue the same conclusion with local papers and
  baseline code, but without access to the other investigations.
  \item[Approximation.] Assume hardness, consult the literature, and seek both
  a provable bicriteria guarantee and a strong empirical frontier.
  \item[Synthesis.] Inspect the accumulated record and build the final solver
  against an evaluator frozen before its implementation.
\end{description}

The two exact conditions provide a useful control.  Despite their different
information, both arrived at essentially the same architecture:
\[
  \begin{gathered}
    \text{blind exact}\quad\parallel\quad\text{informed exact}\\
    \Downarrow\\[-1mm]
    \text{parent-arc master + lazy connectivity/radius cuts}\\
    \Downarrow\\[-1mm]
    \text{same gap: rounds, pivots, and rational bit growth}
  \end{gathered}
\]
Polynomial-time separation of the individual cuts did not give a
polynomial-time exact algorithm.  A concrete scalarization witness was the
three-point frontier
\[
  \begin{aligned}
    \mathcal F&=\{(51,20),\ \boxed{(42,23)},\ (34,24)\},\\
    (42,23)\in\arg\min_{(L,R)\in\mathcal F}(L+\lambda R)
      &\ \Longrightarrow\ \lambda\leq3\ \land\ \lambda\geq8.
  \end{aligned}
\]
Thus the boxed Pareto point is unsupported.  The blind condition found an
analogous witness, and both conditions found dominated local-exchange optima.
We retained the formulations as small-instance oracles and the counterexamples
as boundaries; neither exact condition was reported as a proof that the
problem is in P.

\begin{example}[A counterexample that changed the proof object]
\label{ex:broom-counterexample}
The approximation condition conjectured that the safe minimum of two
\(\alpha=2\) LAST-style endpoint policies always contains a factor-two tree.
The claim survived exhaustive and randomized testing, but an \(n=58\)
subdivided broom has
\[
  \begin{aligned}
    M&=57, & \directradius&=9, & (L,R)&=(117,12),\\
    \max\!\left\{\frac{L}{M},\frac{R}{\directradius}\right\}
      &=\frac{39}{19}>2.
  \end{aligned}
\]
The witness redirected the proof from per-terminal stretch to the global
maximum radius used by the objective.  Charging whole residual components
instead yields
\[
  R(A_H)\leq \directradius+H,\qquad
  \bigl(L(A_H)-M\bigr)H\leq M\directradius .
\]
These are the height-partition inequalities of
Theorem~\ref{thm:height-partition}, not a repaired version of the falsified
invariant.
\end{example}

\paragraph{Synthesis continued the approximation route.}
The final condition did not assemble all preceding solvers.  Its inherited
spine and its additions were:

\begin{description}
  \item[Inherited.] Classical anchors, height and centered partitions, exact
  objective reconstruction, the frozen evaluator, and portfolio comparison.
  \item[Added.] Bounded reparenting and component exchange, cap repair,
  separately budgeted primary and companion search branches,
  deterministic portfolio selection, and size- and geometry-dependent
  dispatch.
  \item[Not imported.] The two rational cutting-plane engines, whose
  uncontrolled exact-search tails conflicted with the production role.
\end{description}

The implementation was largely rewritten, but the line of attack continued:
the \textsc{certified} mode preserves the theorem-bearing construction,
while the empirical modes begin from the same pool and spend additional work
on frontier coverage.

Each condition maintained claim and counterexample ledgers.  Candidate
theorems were converted into explicit inequalities, candidate algorithms were
checked on exact small instances, and negative witnesses were passed forward
instead of discarded.  The author selected the problem and information
boundaries, reconstructed the proofs, checked the literature boundary, fixed
the evaluator, and owns the final claims.

The public repository is a reproducibility subset rather than a transcript.
It releases the five task specifications, final solver and checkers, selected
benchmarks and baseline adaptations, frozen results, and completed Lean
modules.  It does not release complete conversations, private references,
unfinished formalizations, or every exploratory ledger and program.
Appendix~\ref{app:record} gives the exact inventory and formalization boundary.

%% file: theory.tex
\section{New Theory}
\label{sec:theory}

The two theoretical results address different questions.  The first explains
why exact optimization has resisted the sequence of simple constructions
reviewed in Section~\ref{sec:classical}.  The second shows that hardness does
not prevent a particularly simple global cost--radius guarantee.

\subsection{NP-completeness via a weak reduction}

\begin{theorem}[Complexity of rectilinear RCRST]
\label{thm:npc}
The rooted, terminal-only Manhattan cost--radius spanning-tree decision
problem is NP-complete under polynomial-time many-one reductions.  The
reduction establishes weak NP-hardness; it does not establish strong
NP-hardness.
\end{theorem}

Membership in NP follows from an \(n-1\)-edge certificate.  Exact binary
arithmetic gives every edge length and total cost, while one rooted traversal
computes all source distances and the radius.

For hardness, reduce from positive-integer \textsc{Partition}.  Given
\(a_1,\ldots,a_n\), let
\[
  S=\sum_{i=1}^n a_i,\qquad h=10S,\qquad t=4S,
  \qquad X=2hn+t .
\]
Create axis terminals and upper terminals
\[
  s_i=(2hi,0)\quad(0\leq i\leq n),\qquad
  u_i=(2hi-3a_i,a_i)\quad(1\leq i\leq n),
\]
together with \(z=(X,0)\); the root is \(s_0\).  Set
\[
  B=X+3S,\qquad D=X+S .
\]
The strict left-to-right order is
\(s_0,u_1,s_1,\ldots,u_n,s_n,z\), because every long interblock
gap is at least \(17S\).

If a subset \(I\) has sum \(A=S/2\), each selected block uses the path
\(s_{i-1}u_is_i\), while each unselected block uses
\(s_{i-1}s_i\) and attaches \(u_i\) as a leaf.  The resulting tree has
\[
  L=X+4S-2A=B,\qquad
  d_T(s_0,z)=X+2A=D ,
\]
and the off-path upper terminals also lie within radius \(D\).

The reverse implication must rule out every edge of the complete Manhattan
graph, not only the displayed local edges.  Let \(P\) be the
\(s_0\)-to-\(z\) path of an arbitrary feasible tree and put
\[
  F=L(T)+d_T(s_0,z)\leq B+D=2X+4S.                 \tag{1}
\]
Across every vertical gap in the terminal order, connectivity contributes one
tree crossing and \(P\) contributes one path crossing.  Their horizontal
contribution already totals \(2X\).  A third crossing of a nonfinal
interblock gap costs more than the remaining \(4S\); a third crossing of the
final \(4S\) gap leaves no room for the strictly positive vertical cost.
Every interblock prefix cut therefore has one tree crossing, which is also
the unique path crossing.  Nested-cut uniqueness forbids a bridge from
skipping a block, and edge counting forces every internal edge \(u_is_i\).

Let \(W\) be the total bridge length and let \(A\) now denote the sum of
items whose internal edge lies on \(P\).  The forced topology gives
\[
  L(T)=W+4S,\qquad d_T(s_0,z)=W+4A.                \tag{2}
\]
Horizontal and vertical telescoping along \(P\), including internal edges
traversed in either direction, yields the bridge inequality
\[
  W+2A\geq X.                                      \tag{3}
\]
Equations (2)--(3) and \(L(T)\leq B\) imply \(A\geq S/2\), while
(2)--(3) and \(d_T(s_0,z)\leq D\) imply \(A\leq S/2\).
Thus \(A=S/2\), recovering a partition.  Appendix~\ref{app:npc}
gives the complete cut and bit-complexity arguments.

\subsection{A bottom-up height partition}

Root an MST \(T_M\) at \(r\).  For a residual component whose top vertex is
\(v\), let \(h_v\) be its greatest retained tree distance from \(v\).
Algorithm~\ref{alg:height-partition} keeps the parent edge of \(v\) exactly
when doing so preserves residual height at most \(H\).  Otherwise it cuts that
edge and reconnects \(v\) directly to the root.

\begin{algorithm}[t]
  \caption{\textsc{HeightPartition}\((T_M,r,H)\)}
  \label{alg:height-partition}
  \begin{algorithmic}[1]
    \Require Rooted MST \(T_M\), root \(r\), threshold \(H>0\)
    \Ensure A terminal spanning tree \(A_H\)
    \State \(A_H\gets T_M\); \(h_v\gets0\) for every \(v\in V\)
    \ForAll{\(v\in V\setminus\{r\}\) in postorder}
      \State \(u\gets\operatorname{parent}_{T_M}(v)\);
             \(e\gets d(u,v)\)
      \If{\(e+h_v\leq H\)}
        \State \(h_u\gets\max\{h_u,e+h_v\}\)
      \Else
        \State \(A_H\gets A_H-\{(u,v)\}+\{(r,v)\}\)
      \EndIf
    \EndFor
    \State \Return \(A_H\)
  \end{algorithmic}
\end{algorithm}

\begin{theorem}[Global cost--radius tradeoff]
\label{thm:height-partition}
For every real \(H>0\), Algorithm~\ref{alg:height-partition} returns a
terminal tree \(A_H\) satisfying
\[
  R(A_H)\leq\directradius+H,\qquad
  \bigl(L(A_H)-\mstcost\bigr)H\leq\mstcost\directradius.       \tag{4}
\]
It runs in linear time after the MST has been rooted.
\end{theorem}

\begin{proof}[Proof sketch]
Every residual component has height at most \(H\).  The root component
therefore has radius at most \(H\); every other component first uses one
direct root edge of length at most \(\directradius\), proving the radius
bound.

Consider a cut child \(v\), its deleted parent edge of weight \(e_v\), and the
total weight \(W_v\) retained in its residual component.  MST cut optimality
gives
\[
  0\leq d(r,v)-e_v\leq\directradius .
\]
The failed keep test and \(W_v\geq h_v\) give
\(e_v+W_v>H\).  Charge the shortcut's net addition to these edges:
\[
  \bigl(d(r,v)-e_v\bigr)H
     \leq\directradius(e_v+W_v).                   \tag{5}
\]
Once a residual is cut, none of its edges enters an ancestor residual, so the
charged edge sets are disjoint and have total weight at most \(\mstcost\).
Summing (5) proves the length inequality.  Each residual component is
reconnected exactly once, so the result is a tree.  The full argument appears
in Appendix~\ref{app:height}.
\end{proof}

Equivalently,
\[
  \boldsymbol{\rho}(A_H)
  \preceq
  \left(1+\frac{\directradius}{H},
        1+\frac{H}{\directradius}\right),
  \qquad
  \boxed{\boldsymbol{\rho}(A_{\directradius})\preceq(2,2)}.   \tag{6}
\]
The integer implementation samples the continuous envelope in (6).  The
balanced choice \(H=\directradius\) gives a single tree whose cost and radius
are each within two of their independent lower bounds.

\begin{figure*}[t]
  \centering
  \begin{minipage}[c]{0.53\textwidth}
    \centering
    \begin{tikzpicture}[
        x=1cm,y=0.82cm,
        vertex/.style={circle,draw=black!70,fill=white,
                       minimum size=5.2mm,inner sep=0pt,font=\footnotesize},
        root/.style={rectangle,draw=black,fill=black,text=white,
                     minimum size=5.2mm,inner sep=0pt,font=\footnotesize},
        kept/.style={line width=1.05pt,draw=black!72},
        deleted/.style={densely dotted,line width=1pt,draw=FDred!82},
        newedge/.style={line width=1.45pt,draw=FDdark},
        component/.style={draw=black!18,fill=black!2,rounded corners=4pt,
                          inner sep=6pt},
        lab/.style={font=\footnotesize,fill=white,inner sep=1pt}
      ]
      \node[root] (r) at (0,0) {\(r\)};
      \node[vertex] (u) at (0,1.25) {\(u\)};
      \node[vertex] (x) at (-1.15,2.45) {\(x\)};
      \node[vertex] (v) at (1.35,2.35) {\(v\)};
      \node[vertex] (a) at (0.72,3.75) {\(a\)};
      \node[vertex] (b) at (2.05,3.75) {\(b\)};

      \begin{scope}[on background layer]
        \node[component,fit=(r)(u)(x)] {};
        \node[component,fit=(v)(a)(b)] {};
      \end{scope}

      \draw[kept] (r)--node[lab,left] {\(2\)} (u);
      \draw[kept] (u)--node[lab,above left] {\(1\)} (x);
      \draw[kept] (v)--node[lab,left] {\(3\)} (a);
      \draw[kept] (v)--node[lab,right] {\(1\)} (b);
      \draw[deleted] (u)--node[lab,sloped,above,pos=.36,
        text=FDred!90!black] {\(e=2\)} (v);
      \node[font=\footnotesize,text=FDred!90!black] at (0.72,1.78) {\(\times\)};
      \draw[newedge] (r) to[bend right=27]
        node[lab,sloped,below,text=FDdark!82!black]
        {\(d(r,v)\leq\directradius\)} (v);

      \draw[decorate,decoration={brace,amplitude=3pt},black!55]
        (-1.58,-0.02)--node[left=4pt,font=\footnotesize,align=center]
        {residual\\height \(\leq H\)} (-1.58,2.47);
      \draw[decorate,decoration={brace,amplitude=3pt,mirror},black!55]
        (2.48,2.33)--node[right=4pt,font=\footnotesize,align=center]
        {\(h_v=3\leq H\)} (2.48,3.78);

      \node[font=\small\bfseries] at (0.45,4.55)
        {Bottom-up partition (\(H=4\))};
      \node[font=\footnotesize] at (0.45,-0.72)
        {\(e+h_v=2+3=5>H=4\ \Longrightarrow\ (u,v)\mapsto(r,v)\)};
    \end{tikzpicture}
  \end{minipage}\hfill
  \begin{minipage}[c]{0.43\textwidth}
    \centering
    \begin{tikzpicture}
      \begin{axis}[
          width=\linewidth,height=5.05cm,
          xmin=1.18,xmax=3.25,ymin=1.18,ymax=3.25,
          axis lines=left,
          xlabel={\(\bar\ell\) (length bound)},
          ylabel={\(\bar r\) (radius bound)},
          xlabel style={font=\small},
          ylabel style={font=\small},
          tick label style={font=\footnotesize},
          xtick={1.5,2,3},
          ytick={1.5,2,3},
          clip=false
        ]
        \addplot[FDdark,line width=1.2pt,domain=1.33:3.2,samples=120]
          {1+1/(x-1)};
        \addplot[only marks,mark=*,mark size=1.9pt,
          mark options={draw=black!55,fill=black!55}]
          coordinates {(3,1.5) (1.5,3)};
        \addplot[only marks,mark=*,mark size=3pt,
          mark options={draw=FDdark!75!black,fill=FDdark}]
          coordinates {(2,2)};
        \draw[densely dotted,black!35]
          (axis cs:2,1.18)--(axis cs:2,2)--(axis cs:1.18,2);
        \node[font=\footnotesize,anchor=south west,text=FDdark!82!black]
          at (axis cs:2.03,2.03) {\(H=\directradius:\ (2,2)\)};
        \node[font=\footnotesize,anchor=north east] at (axis cs:2.97,1.47)
          {\(H=\directradius/2\)};
        \node[font=\footnotesize,anchor=south west] at (axis cs:1.53,3.02)
          {\(H=2\directradius\)};
        \node[font=\footnotesize,anchor=south west,text=black!55]
          at (axis cs:1.22,1.22) {\((1,1)\) ideal};
        \draw[-{Latex[length=1.5mm]},black!55]
          (axis cs:1.52,1.55)--(axis cs:1.30,1.32);
        \node[font=\small\bfseries] at (axis cs:2.24,3.48)
          {Certified tradeoff};
      \end{axis}
    \end{tikzpicture}
  \end{minipage}
  \caption{Bottom-up height partition and its certified tradeoff.  The dotted
  red edge fails the residual-height test and is replaced by the blue root
  edge (left).  The curve is the coordinatewise upper-bound envelope in
  (6), not an empirical Pareto frontier (right).  Color supplements line
  weight and style; the construction remains legible in grayscale.}
  \Description{Left: a rooted weighted MST is partitioned bottom-up when a
  child residual would exceed height H, and the separated component root is
  joined directly to the source. Right: a hyperbolic certified tradeoff curve
  between normalized length and radius bounds, highlighting the balanced
  two-by-two point.}
  \label{fig:height-tradeoff}
\end{figure*}

\subsection{The factor-two boundary}

Theorem~\ref{thm:height-partition} uses only metricity and MST optimality.
That level of generality cannot yield a universal common factor below two for
a single returned tree: for every \(c<2\), an explicit finite metric admits
no tree \(T\) with both \(L(T)\leq c\mstcost\) and
\(R(T)\leq c\directradius\).  Appendix~\ref{app:height} gives the
trunk-and-arms construction.

This boundary is deliberately narrow.  The construction is not embedded in
the Manhattan plane, and it does not rule out a bounded portfolio whose
different members cover different comparison trees.  A below-two planar
\(L_1\) guarantee and a below-two approximate-Pareto portfolio theorem remain
open.

%% file: algorithm.tex
\section{\hprcrst: A Certified Core and Empirical Frontiers}
\label{sec:algorithm}

The theory suggests a separation that is useful in practice.  One small
construction family can carry a universal certificate; a larger search can
then improve the measured frontier without being asked to justify that
certificate.  \hprcrst\ makes this separation explicit.  It accepts a rooted
terminal set and returns a bounded portfolio of parent arrays, rather than a
single tree for one preselected scalarization.

Its algorithmic spine comes from the preceding approximation investigation:
classical shallow-light anchors (trees returned directly by named
constructions), the height-partition and centered families, exact
remeasurement, and a portfolio objective.  The final implementation rewrites
that spine and extends it with bounded reparenting, component exchange, cap
repair, and size-dependent dispatch.  Its independent search branches each
have their own seed rule, operators, and finite budget.  The earlier exact
cutting-plane engines are not part of the production solver.

\subsection{One exact boundary for every candidate}

A candidate is represented by one parent for each nonroot terminal.  Before a
candidate may enter a portfolio, a common validation path checks that the
parent relation is connected and acyclic, recomputes every Manhattan edge and
root distance, and obtains \(L\) and \(R\) from the reconstructed tree.
Incremental move formulas are used to propose trees, never to certify their
reported objectives.  Metric duplicates and dominated points are then
removed with exact integer comparisons.

This boundary serves two purposes.  It prevents a fast local update from
silently changing the problem being measured, and it lets heterogeneous
constructions be combined without trusting their internal bookkeeping.
Products used by certificates and ratio comparisons are evaluated with
unbounded integers; the supported signed-coordinate domain and overflow
argument appear in Appendix~\ref{app:algorithm}.

\paragraph{Five implementation terms.}
An \emph{anchor} is a tree built directly by a named construction, such as
MST, star, or LAST.  A \emph{seed} is a validated candidate handed to a local
improvement operator.  A \emph{search branch} is one choice of seed rule,
operators, and finite budget.  In the centered fixed-MST family, each
component keeps its MST edges and a nonroot component attaches directly from
the root to one chosen center.  The primary branch enables that family; the
companion branch disables it and keeps a separate budget.  Their validated
outputs are united only after both branches finish.  The nondominated result
is a \emph{frontier}; the returned
\emph{portfolio} contains at most \(k\) selected trees from that frontier.

\begin{algorithm}[t]
  \caption{\textsc{CertifiedCore}\((V,r)\)}
  \label{alg:hprcrst-core}
  \begin{algorithmic}[1]
    \Require Terminals \(V\), root \(r\)
    \State \(T_M\gets\textsc{DensePrim}(V)\);
           \(\directradius\gets\max_{v\in V}d(r,v)\)
    \State \(A_{\directradius}\gets
      \textsc{HeightPartition}(T_M,r,\directradius)\)
      \Comment{Theorem~\ref{thm:height-partition}}
    \State \(\mathcal C\gets
      \{T_M,\textsc{Star}(V,r),A_{\directradius}\}\)
    \State \Return
      \((\mathcal C,T_M,\directradius,A_{\directradius})\)
  \end{algorithmic}
\end{algorithm}

\begin{algorithm}[t]
  \caption{\textsc{HP-RCRST}\((V,r,\mu,k)\): mode extension and selection}
  \label{alg:hprcrst}
  \begin{algorithmic}[1]
    \Require Terminals \(V\), root \(r\), mode \(\mu\), portfolio cap \(k\)
    \State \((\mathcal C,T_M,\directradius,A_{\directradius})
      \gets\textsc{CertifiedCore}(V,r)\)
    \If{\(\mu=\textsc{certified}\)}
      \State \(\mathcal C\gets\mathcal C\cup
        \{\textsc{HeightPartition}(T_M,r,H):
          H\in\{\directradius/4,\directradius/2,
                 2\directradius,4\directradius\}\}\)
    \Else
      \State \(\mathcal C_0\gets\mathcal C\cup
        \textsc{ClassicalAnchors}(V,r,\mu)\)
      \ForAll{independently budgeted search branches \(B\) enabled by \(\mu\)}
        \State \(\mathcal C_B\gets\mathcal C_0\)
        \If{\(\textsc{CenteredEnabled}(B)\)}
          \State \(\mathcal C_B\gets\mathcal C_B\cup
            \textsc{CenteredPartitions}(T_M,r,\mu)\)
        \EndIf
        \If{\(\mu\in\{\textsc{balanced},\textsc{quality}\}\)
            and \(|V|\leq8\)}
          \State \(\mathcal C_B\gets\mathcal C_B\cup
            \textsc{ExactTreeFrontier}(V,r)\)
        \EndIf
        \State \(\mathcal S\gets\textsc{SelectSeeds}(\mathcal C_B,B)\)
        \State \(\mathcal C_B\gets\mathcal C_B\cup
          \textsc{Improve}(\mathcal S,B)\)
        \State \(\mathcal C\gets\mathcal C\cup
          \textsc{Nondominate}(\textsc{Validate}(\mathcal C_B))\)
      \EndFor
    \EndIf
    \State \(\mathcal F\gets
      \textsc{Nondominate}(\textsc{Validate}(\mathcal C))\)
    \If{\(\mu=\textsc{certified}\)}
      \State \Return
        \(\textsc{CertifiedSelect}(\mathcal F,A_{\directradius},k)\)
    \Else
      \State \Return \(\textsc{PortfolioSelect}(\mathcal F,k,\mu)\)
    \EndIf
  \end{algorithmic}
\end{algorithm}

Algorithms~\ref{alg:hprcrst-core} and~\ref{alg:hprcrst} separate the shared
theorem-derived core from mode-specific work.  The source contains more
constructions than the pseudocode names individually, but every one enters
through the same validate--nondominate boundary.  The exact schedules and
cutoffs are recorded in Appendix~\ref{app:algorithm} and in the released
artifact.

\subsection{The certified portfolio}

The \textsc{certified} mode constructs a canonical dense-Prim MST, the
star, and height partitions at five logarithmically spaced thresholds.  The
\(H=\directradius\) member is essential.  If another returned tree
Pareto-dominates it, the dominator may carry the certificate; portfolio
truncation is otherwise forbidden to remove the witness.  With \(k=1\), the
exact \(H=\directradius\) parent array is returned.

The dense MST costs \(O(n^2)\) time and \(O(n)\) auxiliary memory, and each
postorder partition is linear.  Thus the complete certified portfolio remains
deterministic and polynomial:
\[
  \exists T\in\mathcal P_{\rm cert}:
  \qquad L(T)\leq2\mstcost,\qquad R(T)\leq2\directradius.       \tag{7}
\]
Immediately before serialization, \hprcrst\ independently rebuilds the MST
and the balanced partition and rechecks both inequalities in (7).

\subsection{Exact small instances and bounded improvement}

The empirical modes begin with several deliberately different candidate
families.
Classical star, MST, PD, BRBC, and KRY/LAST constructions cover recognizable
endpoint and shallow-light behaviors.  Radial-label, cap-repair, and
critical-height constructions add topologies not obtained by a single
Prim--Dijkstra parameter sweep.  These anchors are useful both as returned
points and as starting trees for later moves.

Two exact mechanisms are retained only on bounded domains.  The first is
\textsc{ExactTreeFrontier} for \(n\leq8\).  A Pr\"ufer code is a sequence of
\(n-2\) terminal indices, with repetition allowed.  It is in one-to-one
correspondence with the unrooted labeled trees on those \(n\) terminals.
The routine enumerates all \(n^{n-2}=2^{\Theta(n\log n)}\) codes, decodes each
tree, orients it at \(r\), remeasures it, and retains the nondominated pairs.
This is exactly the terminal-only feasible domain; trees with added Steiner
vertices are outside both the enumeration and our problem definition.
For example, on terminals \(\{0,1,2,3\}\),
\((1,1)\leftrightarrow\{(1,0),(1,2),(1,3)\}\), the star centered at terminal
\(1\).  The growth rate explains the small cutoff.

The second exact mechanism is a polynomial dynamic program with quadratic
memory.  Fix the rooted MST \(T_M\) and a radius cap \(D\).
Cut MST edges into connected components; retain intracomponent edges, keep the
root component centered at \(r\), and attach every other component \(C\) by
one edge \((r,c)\) to a chosen terminal \(c\in C\).  Feasibility requires
\[
  d(r,c)+\max_{v\in C}d_{T_M}(c,v)\leq D.                       \tag{8}
\]
If \(F[u,c]\) is the minimum net change inside the rooted subtree \(T_u\)
while \(u\) remains in an open component centered at \(c\), and
\[
  G[v]=\min_{z\in T_v}\{d(r,z)+F[v,z]\},
\]
then a child \(v\) containing \(c\) must be kept.  Every other child
contributes
\[
  \min\{F[v,c],\,G[v]-d(u,v)\}.                                \tag{9}
\]
The \(O(n^2)\) states and transitions enumerate exactly the feasible
fixed-MST centered family.  Reconstruction gives the minimum-length member
under cap \(D\), whose radius is rechecked.  This is an exact construction
within a restricted family, not a stronger universal approximation theorem.

The bounded search operators have concrete topological meanings:

\begin{description}
  \item[Reparent.] Cut the parent edge of a rooted subtree and attach its root
        to a vertex outside that subtree.
  \item[Component exchange.] Delete one tree edge and reconnect the two
        components through any endpoint pair, allowing the detached component
        to be reoriented.
  \item[Cap repair.] Under a fixed requirement \(R\leq D\), accept only a
        modification that lowers \(L\).
  \item[Depth compression.] Using the old root-distance labels, choose shorter
        valid parent edges while checking that neither \(L\) nor \(R\)
        increases.
  \item[Perturb and descend.] Make a deterministic, coordinate-seeded
        reparenting perturbation, then resume greedy improvement.
\end{description}

A greedy move is accepted only when reconstruction improves its specified
exact score.  Total-detour exchange uses
\(L+\sum_v d_T(r,v)\), not the maximum radius \(R\), as an auxiliary score;
the initial perturbation may deliberately worsen its seed.  Every resulting
tree is nevertheless remeasured and Pareto-filtered in exact \((L,R)\).
The most expensive complete component pass is cubic and is confined to the
small intensive tier.

\subsection{One certified mode, three empirical modes}

Every mode constructs and exactly remeasures the balanced height partition
\(A_{\directradius}\).  The certified selector reserves that tree or an exact
Pareto dominator.  An empirical selector may trade it away when a small
portfolio cap makes frontier coverage more valuable.  Consequently, only
\textsc{certified} carries the universal guarantee in (7); the other
three modes continue from the same construction pool but make empirical
runtime--quality claims.

\begin{description}
  \item[\textsc{fast}] uses expanded deterministic anchors, centered
  candidates and bounded repair, while keeping the search shallow.
  \item[\textsc{balanced}] spends a larger fixed budget on the primary search
  branch for small instances and reverts to the scalable policy above that
  tier.
  \item[\textsc{quality}] adds the full intensive branch, a separately
  budgeted companion branch on selected moderate-size instances, richer
  cap schedules, and exact labeled-tree enumeration at the small cutoff.
\end{description}

On selected moderate-size instances, the companion branch reruns the same
improvement machinery with centered constructions disabled.  This prevents a
strong centered seed from consuming the finite seed and beam budget before a
structurally different topology is even explored.
Finally, the exact frontier is ordered by \((L,R,\mathrm{parent})\).  If more
than \(k\) trees remain, deterministic selection preserves the two endpoints,
and admits marked lower-mode witnesses while capacity remains, then fills the
largest exact area gaps.  Only \textsc{CertifiedSelect} reserves the
factor-two witness unconditionally.  The result is a portfolio whose modes
differ in work, not merely in labels.

%% file: experiments.tex
\section{Experimental Evaluation}
\label{sec:experiments}

The experiments ask three different questions: whether the practical
portfolio improves recent methods, how much quality is lost when only the
certified construction family is used, and whether the four modes produce a
stable runtime--quality ordering.  These questions require different
statistics.  A good tree near the diagonal need not cover an entire
cost--radius frontier.

\subsection{Data, baselines, and metrics}

The frozen comparison contains 28 development instances with \(8\) to \(32\)
terminals, including synthetic, grid, near-line, and extracted industrial
nets.  Its published-method reference is the nondominated union of classical
constructions, the author PD-II implementation swept over 19 parameter
values, and terminal-only trees exported from the 2023 MSPD/MSS
implementation over its prescribed \(4\times4\) settings.  The internal
predecessor developed in our approximation investigation is excluded from
this reference and is compared separately in Appendix~\ref{app:experiments}.
Because these 28 instances were visible during development, we use them for
regression and direct baseline comparison, not as evidence of out-of-sample
generalization.

A separate 50-instance set, not used to configure the final mode ordering,
contains uniform, clustered, grid, comb, radial, and near-line geometries
through \(n=96\).  Scaling uses 30 instances at
\(n\in\{64,128,256,512,1024,2048\}\), with three retained runs per
mode and instance.  Four instances through \(n=8\) are small enough for exact
labeled-tree enumeration via Pr\"ufer codes.

Every output is independently reconstructed from its parent array and
remeasured with integer arithmetic.  For portfolios \(P\) and \(Q\), write
\[
  P\succeq Q
  \quad\Longleftrightarrow\quad
  \forall q\in Q\;\exists p\in P:
  L(p)\leq L(q)\ \land\ R(p)\leq R(q).
\]
We report \emph{dominance} if \(P\succeq Q\) but \(Q\nsucceq P\), a
\emph{tie} if both hold, \emph{mixed} if neither holds, and a \emph{loss} if
only \(Q\succeq P\).  Hypervolume uses the same recorded reference corner and
the same \(M,\directradius\) normalization for both portfolios on each
instance.

To measure only the best balanced representative, define the empirical
portfolio common factor
\begin{equation}
  c_{\rm port}(I,P)=\min_{T\in P}
  \max\left\{\frac{L(T)}{M_I},
             \frac{R(T)}{\Delta_I}\right\}.                     \tag{10}
  \label{eq:empirical-factor}
\end{equation}
Unlike hypervolume, \(c_{\rm port}\) deliberately ignores the breadth of a
portfolio.

\begin{figure*}[t]
  \centering
  \begin{minipage}[c]{0.47\textwidth}
    \centering
    \begin{tikzpicture}
      \begin{axis}[
          width=\linewidth,height=5.15cm,
          xbar stacked,
          bar width=8pt,
          xmin=0,xmax=28,
          xtick={0,7,14,21,28},
          xlabel={instances},
          symbolic y coords={quality,balanced,fast,certified},
          ytick=data,
          yticklabels={quality,balanced,fast,certified},
          tick label style={font=\footnotesize},
          label style={font=\small},
          legend style={
            at={(0.5,-0.29)},anchor=north,legend columns=4,
            draw=none,font=\footnotesize,column sep=3pt
          },
          axis line style={black!45},
          xmajorgrids,
          grid style={black!8}
        ]
        \addplot[fill=FDdark!76,draw=FDdark!88!black]
          coordinates {(23,quality) (23,balanced) (12,fast) (0,certified)};
        \addplot[fill=black!17,draw=black!42]
          coordinates {(5,quality) (5,balanced) (6,fast) (1,certified)};
        \addplot[fill=black!43,draw=black!58]
          coordinates {(0,quality) (0,balanced) (9,fast) (1,certified)};
        \addplot[fill=FDred!70,draw=FDred!88!black]
          coordinates {(0,quality) (0,balanced) (1,fast) (26,certified)};
        \legend{dominance,tie,mixed,loss}
      \end{axis}
    \end{tikzpicture}
  \end{minipage}\hfill
  \begin{minipage}[c]{0.47\textwidth}
    \centering
    \begin{tikzpicture}
      \begin{axis}[
          width=\linewidth,height=5.15cm,
          xmin=0.98,xmax=1.56,ymin=0.98,ymax=1.56,
          xlabel={\(c_{\rm port}(I,\mathcal P_{\rm cert})\)},
          ylabel={\(c_{\rm port}(I,\mathcal P_{\rm ref})\)},
          tick label style={font=\footnotesize},
          label style={font=\small},
          grid=both,
          grid style={black!7},
          axis line style={black!45}
        ]
        \addplot[black!40,dashed,domain=0.98:1.56] {x};
        \addplot[only marks,mark=*,mark size=1.75pt,
          mark options={fill=FDdark!60,draw=FDdark!90!black}]
          coordinates {
            (1.260000,1.043956) (1.310638,1.076923)
            (1.006716,1.006716) (1.264129,1.149956)
            (1.310986,1.026114) (1.158148,1.025426)
            (1.270175,1.183178) (1.341369,1.053264)
            (1.431338,1.095951) (1.396495,1.061336)
            (1.301782,1.037347) (1.091743,1.018349)
            (1.335148,1.157764) (1.478173,1.123492)
            (1.320775,1.120666) (1.240458,1.140983)
            (1.111894,1.047493) (1.528837,1.086608)
            (1.047001,1.019376) (1.056304,1.025704)
            (1.014953,1.012461) (1.037486,1.018454)
            (1.056772,1.023520) (1.109697,1.057576)
            (1.000000,1.000000) (1.237981,1.000000)
            (1.077419,1.074074) (1.083499,1.071571)
          };
        \draw[densely dotted,FDdark!80]
          (axis cs:1.209997,0.98)--(axis cs:1.209997,1.56);
        \draw[densely dotted,black!55]
          (axis cs:0.98,1.062795)--(axis cs:1.56,1.062795);
        \node[font=\footnotesize,anchor=south west,text=FDdark!85!black]
          at (axis cs:1.215,1.50) {certified mean \(1.21\)};
        \node[font=\footnotesize,anchor=south west,text=black!60]
          at (axis cs:1.285,1.067) {reference mean \(1.06\)};
        \node[font=\footnotesize,anchor=south east,text=black!50]
          at (axis cs:1.54,1.52) {equal factor};
      \end{axis}
    \end{tikzpicture}
  \end{minipage}
  \caption{Two views of the 28-instance comparison with the published-method
  union.  The stronger practical modes cover the full union, whereas the
  theorem-first portfolio does not (left).  Nevertheless, its best balanced
  tree lies well inside the factor-two guarantee; each point at right compares
  its empirical common factor with that of the denser published portfolio.
  The panels measure frontier breadth and one-point balance, respectively.}
  \Description{Left: stacked horizontal bars give dominance, tie, mixed, and
  loss counts for four modes against the published-method union.  Quality and
  balanced have twenty-three dominances and five ties; certified has
  twenty-six losses, one mixed case, and one tie. Right: twenty-eight scatter
  points compare the best common factor of the certified portfolio with the
  reference portfolio.}
  \label{fig:quality-certificate}
\end{figure*}

\subsection{Comparison with the frozen union}

Figure~\ref{fig:quality-certificate}a separates the theorem-first mode from
the empirical portfolio modes.  Balanced and quality never lose to the
published-method reference: each strictly covers it on 23 instances and ties
on the remaining five.  Their sums of normalized per-instance hypervolume
differences are \(0.0863\) and \(0.0863\), respectively.  The small difference is
consistent with their intended roles: quality spends additional work to
recover a few frontier points, rather than changing the objective.  Fast is
also positive in aggregate hypervolume (\(0.0654\)), with 12 dominances, six
ties, nine mixed cases, and one loss.

The certified portfolio is covered by the published reference on 26
instances, is mixed on one, and ties on one; its aggregate hypervolume
difference is \(-0.7478\).
The approximation theorem is therefore not a frontier-dominance result.
In direct comparisons, balanced and quality each dominate PD-II on 25
instances and tie it on three; each dominates the terminal MSPD/MSS portfolio
on 23 and ties it on five.  The union is the stricter comparison because a
candidate must cover the complementary points supplied by all published
methods at once.

\subsection{What does the certificate cost?}
\label{sec:certified-quality}

The right panel of Figure~\ref{fig:quality-certificate} uses
Eq.~\eqref{eq:empirical-factor}.  On the 28 development instances, the
certified portfolio has mean \(c_{\rm port}=1.21\), median \(1.24\), and range
\(1.00\)--\(1.53\).  Its nondominated output contains one to six trees, with
median three.  The much denser reference union has mean \(1.06\) and maximum
\(1.18\).

The apparent tension with the left panel is informative.  A five-threshold
height sweep often finds a respectable point near the diagonal, far inside
its universal ceiling of two.  It does not, however, populate the shoulders
between the MST and shortest-path endpoints.  Empirical search improves the
central compromise and, more decisively, fills those missing regions.  The
theorem construction is thus a practical core, but not a substitute for the
frontier engine.

The same distinction remains on held-out data.  There the certified
portfolio has mean \(c_{\rm port}=1.25\), median \(1.22\), and maximum
\(1.97\).
The factor-two witness is independently reconstructed on all 78 development
and held-out instances: 66 portfolios contain the original balanced
height-partition tree and 12 contain an exactly remeasured Pareto dominator.

\subsection{Mode ordering and scaling}

On the 50 held-out instances, balanced strictly covers fast on 21 and ties it
on 29; quality strictly covers balanced on nine and ties it on 41.  Fast
strictly covers the certified portfolio on 44, ties it on three, and is
incomparable on three; it never loses.  Additional empirical budget therefore
gives a monotone portfolio hierarchy, while \textsc{certified} remains a
separate theorem-first endpoint.

\begin{figure}[t]
  \centering
  \begin{tikzpicture}
    \begin{axis}[
        width=\columnwidth,height=5.2cm,
        xmode=log,log basis x=2,
        ymode=log,log basis y=10,
        xmin=60,xmax=2200,ymin=5,ymax=400,
        xtick={64,128,256,512,1024,2048},
        xticklabels={64,128,256,512,1024,2048},
        xlabel={terminals \(n\)},
        ylabel={median time (ms)},
        tick label style={font=\footnotesize},
        label style={font=\small},
        grid=both,
        grid style={black!8},
        legend style={draw=none,font=\footnotesize,at={(0.02,0.98)},
                      anchor=north west}
      ]
      \addplot[FDdark,line width=1.25pt,mark=*]
        coordinates {(64,6.137) (128,6.273) (256,6.135)
                     (512,9.350) (1024,13.275) (2048,27.141)};
      \addplot[black!75,line width=1.05pt,mark=square*]
        coordinates {(64,8.672) (128,16.983) (256,15.312)
                     (512,18.096) (1024,42.204) (2048,140.822)};
      \addplot[black!45,line width=1.05pt,dashed,mark=triangle*]
        coordinates {(64,8.714) (128,16.900) (256,15.560)
                     (512,18.548) (1024,41.941) (2048,142.153)};
      \addplot[FDred!85!black,line width=1.1pt,dashdotted,mark=diamond*]
        coordinates {(64,9.807) (128,19.915) (256,26.237)
                     (512,22.660) (1024,54.241) (2048,182.072)};
      \legend{certified,fast,balanced,quality}
    \end{axis}
  \end{tikzpicture}
  \caption{Single-thread end-to-end median time on the scaling set.  Each
  point aggregates three runs for each of five geometries; process startup
  and serialization are included.}
  \Description{A log-log line plot of median runtime from 64 to 2048
  terminals. Certified is lowest throughout and reaches 27 milliseconds
  at 2048; the three empirical modes reach about 141 to 182 milliseconds.}
  \label{fig:scaling}
\end{figure}
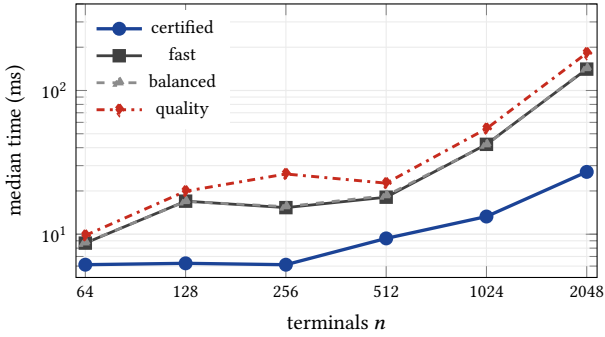

All 360 scaling invocations completed.  At \(n=2048\):
\begin{center}
  \small
  \begin{tabular}{lrrr}
    \toprule
    mode & median ms & max ms & peak MiB\\
    \midrule
    certified & 27.1 & 32.5 & 2.48\\
    fast      & 140.8 & 182.8 & 3.66\\
    balanced  & 142.2 & 220.4 & 3.75\\
    quality   & 182.1 & 302.3 & 4.05\\
    \bottomrule
  \end{tabular}
\end{center}
The source is built with portable \texttt{-O3 -DNDEBUG} settings and uses one
thread; no architecture-specific compiler flag is enabled.

On the held-out set, the certified mode takes a median \(6.9\) ms.  Fast,
balanced, and quality take \(13.0\), \(62.2\), and \(63.3\) ms.  These
are absolute candidate times.  The published artifacts for the recent
baselines do not contain
same-machine timings under the common terminal-only objective, so we do not
form a runtime ratio from their paper-reported numbers.

\subsection{Exactness and component evidence}

Quality reproduces the complete exact frontier on all four small instances
for which every labeled tree was enumerated.  The release also passes 26 unit
and property tests, independent certificate reconstruction, and exact
remeasurement of every reported tree.

\begin{table}[t]
  \caption{Compact quality-mode ablation on the 28-instance frozen union.}
  \label{tab:ablation}
  \centering
  \small
  \begin{tabular}{lrrrrr}
    \toprule
    Variant & Dom. & Tie & Mix. & Loss & \(\sum\Delta\mathrm{HV}\)\\
    \midrule
    construction only     & 1  & 9 & 8 & 10 & \(-0.0589\)\\
    scalable tier only   & 12 & 7 & 9 & 0  & \(+0.0665\)\\
    no deep search        & 20 & 6 & 2 & 0  & \(+0.0840\)\\
    no component exchange & 23 & 5 & 0 & 0  & \(+0.0857\)\\
    reduced starts        & 23 & 5 & 0 & 0  & \(+0.0862\)\\
    \hprcrst\ quality     & 23 & 5 & 0 & 0  & \(\mathbf{+0.0863}\)\\
    \bottomrule
  \end{tabular}
\end{table}

Table~\ref{tab:ablation} confirms that constructions alone do not explain the
full result: they fail to cover the published union on 18 instances.  The
scalable-only policy makes the largest single jump, and deep search removes the
remaining losses.  Component exchange and additional starts then recover
smaller complementary regions.  The ablations are not additive, because
constructions are also search seeds.  We return to that interaction in
Section~\ref{sec:discussion}.

%% file: discussion.tex
\section{Discussion}
\label{sec:discussion}

\subsection{What is settled, and what is not}

The reduction settles the original terminal-only Manhattan decision problem only
at the level of ordinary NP-completeness.  Because its coordinates scale with
the numerical \textsc{Partition} instance, it leaves strong hardness and
pseudopolynomial exact algorithms open.  The approximation theorem is
similarly specific.  It gives one tree with
\(\boldsymbol{\rho}\preceq(2,2)\) in every metric instance, and the
arbitrary-metric construction shows that this line of reasoning cannot force
a smaller universal common factor for one tree.  It does not rule out a
below-two theorem using planar \(L_1\) structure, nor an approximate-Pareto
guarantee in which different portfolio members cover different comparison
trees.

These boundaries matter because exact optimization, geometric approximation,
and portfolio design are different questions.  Earlier routing work often met
all three through the same construction; the present results separate them.

\subsection{Why two exact routes converged}

The two polynomial-time investigations received sharply different
information.  One saw only the problem, benchmark, and evaluator; the other
saw the accumulated references and base code.  Nevertheless, both reached a
parent-arc master with exact rational arithmetic, lazy rooted-connectivity
cuts, and path or potential cuts for radius violations.

The convergence was structural rather than social.  Once a terminal tree is
rooted, its combinatorial content is exactly one parent choice for each
nonroot terminal, so total length is linear in the selected parent arcs.
Radius has an equally canonical local form: a selected arc adds a
nonnegative detour increment to the child's root-path potential.  The two
remaining global failures are root disconnection and excessive detour on a
parent chain.  The first exposes a minimum-cut certificate; the second
exposes a forbidden path, projected potential, or flow inequality.  The
shared architecture is therefore the decomposition induced by the problem.

\begin{example}[Independent witnesses to the same obstruction]
\label{ex:independent-scalarization}
The blind condition found a five-terminal frontier containing the neighboring
points
\[
  (L,R)\in\{(27,12),\ \boxed{(26,13)},\ (22,15)\}.
\]
For the boxed point to minimize \(L+\lambda R\), comparison with the other two
points requires both \(\lambda\leq1\) and \(\lambda\geq2\).
Independently, the full-context condition found
\[
  (L,R)\in\{(51,20),\ \boxed{(42,23)},\ (34,24)\},
\]
whose boxed point requires both \(\lambda\leq3\) and \(\lambda\geq8\).
The instances and incompatible intervals are different.  What recurs is the
same geometric obstruction: a rectangularly feasible Pareto point can lie
above the lower convex envelope visible to every weighted sum.
\end{example}

The common stopping point has the same explanation.  Rooted-subset,
overlong-path, conflict, and tightened potential inequalities all left small
fractional gaps.  General integer cuts can force finite exact progress, but
polynomial separation in one round does not bound cutting-plane rank, the
total number of pivots, or the bit width accumulated across rounds.  The two
conditions did not independently prove polynomiality; they independently
located the same global integrality defect.

Their negative results were still algorithmically selective.  Scalarization
survived because it cheaply locates supported frontier trees, although
unsupported points rule it out as a complete decision method.  Local
exchange survived because each accepted move improves a checked incumbent,
although explicit dominated local optima rule it out as a proof of
completeness.  The cutting-plane engines themselves were not imported into
\hprcrst.  They supplied formulations, audit targets, and warnings about
uncontrolled exact-search tails.

\subsection{Why the counterexample changed the accounting unit}

The approximation investigation first sought an individual-terminal
\((2,2)\) invariant for threshold and LAST-like traversals.  It survived more
than half a million exhaustive, random, and structured tests before failing
on a 58-terminal subdivided broom.  Subdivision preserves the long route but
splits its removable weight among many small terminal edges.  A root shortcut
then replaces only one small edge, so charges that looked local on coarse
instances accumulate without a global budget.

The height partition changes the unit of accounting.  A shortcut is charged
to the entire residual component that caused the height threshold to fail;
different shortcuts receive disjoint residuals.  The proof controls the
maximum global radius actually used by the objective, rather than promising
the stronger per-terminal stretch statement that the broom falsifies.  The
counterexample in Example~\ref{ex:broom-counterexample} did not merely change
a constant.  It identified the wrong object being amortized.

\subsection{A good central tree is not a good frontier}

The certified experiment gives a more nuanced answer than either ``the theorem
won'' or ``the theorem was only decorative.''  On the 28 development instances,
the small certified portfolio has
\[
  \operatorname{mean}c_{\rm port}=1.21,\qquad
  \max c_{\rm port}=1.53,
\]
where \(c_{\rm port}\) is the best common factor in the portfolio.  On the 50
held-out instances the corresponding values are \(1.25\) and \(1.97\).
Thus the observed balanced point is usually well inside the worst-case
ceiling of two.

The published-method reference nevertheless covers the certified portfolio on
26 of 28 instances; one case is mixed and one is a tie.  This is not a
contradiction.  The common factor asks for one point near the diagonal;
hypervolume and portfolio coverage reward a sequence of useful points from
the MST end to the shortest-path end.  A rooted MST plus a handful of height
thresholds is naturally good at the first task.  Each threshold resets
exactly when accumulated residual height crosses \(H\), so the sweep samples
a one-dimensional structural envelope already aligned with length and
radius.  The same one-dimensional family has too few topological degrees of
freedom to fill both shoulders of an empirical frontier.

\begin{example}[The same MST length, a different rooted geometry]
\label{ex:grid-mst-tie}
On \texttt{grid\_5x5}, \(M=416\) and \(\directradius=72\).  The certified
fixed-MST sweep returns
\[
  (416,140),\qquad (515,89),\qquad (763,72),
\]
whose best common factor is \(1.24\).  The published reference and all three
empirical modes instead contain \((416,72)=(M,\directradius)\), which dominates
the entire certified portfolio.  Several MSTs have length \(416\), but their
rooted radii differ.  Cutting and reattaching one canonical MST explores
thresholds on that topology; it does not explore alternative MST tie
structures.  Here a topology change, not a finer height sweep, closes the
frontier in this extreme tie case.
\end{example}

This also clarifies what search contributes.  It does not rescue a collapsed
certificate: the certified balanced point is already useful and inexpensive.
But it is more than a cosmetic refinement.  Constructions alone fail to cover
the published reference on 18 of 28 instances, whereas the final quality mode
has no loss or mixed case and 23 strict dominances.  Search improves the
central compromise and, more decisively, articulates the rest of the selectable
tradeoff curve.  The height construction is therefore the algorithmic spine
inherited by the final solver.  The certified mode reserves its witness; the
empirical modes begin from the same construction pool and add search, but may
trade that witness away during portfolio truncation.

The certificate is also inexpensive inside the released solver.  Once the
dense MST is available, the certified mode performs only five linear postorder
partitions and exact checks.  At \(n=2048\), its median end-to-end time is
\(27\) ms, compared with \(141\)--\(182\) ms for the three empirical modes.
This is an internal mode comparison, not a cross-paper runtime claim; the
published artifacts do not provide timings under the same objective and
measurement boundary.

\subsection{Why the final solver kept some ideas and not others}

The final condition could inspect every earlier result, but full context did
not produce a union of all earlier algorithms.  A component survived when it
could be assigned a bounded and verifiable role:

\begin{itemize}
  \item the hardness reduction fixed the claim boundary but supplied no
        runtime mechanism;
  \item height partitioning was cheap enough to become a separately checked
        certified mode;
  \item star, MST, PD, BRBC, and KRY/LAST constructions became endpoints and
        diverse search seeds, as well as regression witnesses;
  \item exactness was retained only where its work was explicitly bounded:
        complete labeled-tree enumeration for \(n\leq8\) and the fixed-MST
        centered dynamic program through its deployment cutoff;
  \item component exchange, reparenting, cap repair, and compression were
        retained only after empirically checked improvements;
  \item the earlier rational cutting-plane engines were not ported, because
        their unresolved tails conflicted with the production role.
\end{itemize}

The most important transfer from the exact investigations was therefore
negative knowledge: which scalarizations miss feasible rectangles, which
local minima support no completeness claim, and where rational exact search
loses a polynomial bound.  Broader information narrowed claims and deployment
domains at least as much as it supplied constructions.

\begin{example}[Search changed the trees, not only their number]
\label{ex:uniform-search}
On \texttt{uniform\_n30\_s2}, construction alone produces ten nondominated
trees; quality mode produces 33.  The difference is not merely denser
sampling of one curve.  At radius \(1967\), checked search reduces length
from \(5358\) to \(5179\); at radius \(1351\), it reduces length from \(6370\)
to \(5510\).  Both portfolios begin at the same \((5115,3425)\) MST endpoint.
The construction pool therefore supplies a reachable region and useful
anchors, while exchange search changes the topologies available inside that
region.
\end{example}

\subsection{Why the components form an ecology}

The ablations are non-additive because candidate families participate in a
search process, not merely in a final union.  With a finite seed and
deep-search budget, removing one family also changes the neighborhoods later
operators are allowed to explore.

\begin{example}[A strong family can hide complementary trees]
\label{ex:centered-companion}
For compactness, abbreviate the three instances by
\[
  \begin{aligned}
    U_0&:=\texttt{uniform\_n30\_s0},\\
    U_2&:=\texttt{uniform\_n30\_s2},\\
    C_1&:=\texttt{cluster2\_n30\_s1}.
  \end{aligned}
\]
After exact remeasurement, the companion search without centered candidates
retained eight points absent from the primary branch:
\[
  \begin{aligned}
    U_0:\;&(4330,1880),(4412,1757),(4453,1731),\\
    U_2:\;&(5127,3329),\\
    C_1:\;&(2579,1676),(2582,1674),\\[-1mm]
          &(2594,1666),(2595,1665).
  \end{aligned}
\]
Omitting both centered and critical-height constructions still lost one prior
witness.  The final policy therefore gives the two searches separate
budgets and merges them only after both have produced checked frontiers.
Mixing all seeds before search would recreate the crowding that exposed these
eight points.
\end{example}

For the same reason, provenance tags are terminal labels rather than causal
weights.  Of 263 public-set output trees, 203 are tagged as exchange search,
whereas only 17 retain a PD, LAST, or BRBC tag and seven retain a structural
construction tag.  An anchor may disappear after several exchanges while
remaining necessary to reach the final topology.  Conversely, deleting a
frequently returned family can free enough budget to improve another branch.
The pipeline is better understood through interactions among seeds,
operators, and finite budgets than through a ranking of standalone
heuristics.

\subsection{The research object was the transition}

The most revealing observations in this study are not well described by a
division between ``human'' and ``machine'' work.  The blind and full-context
conditions converged because the formulation forced them to; the
approximation condition advanced when a long-lived conjecture was falsified;
and the final solver used earlier failures more directly than it used earlier
code.  In each case, progress occurred at a transition between inspectable
objects:
\[
  \begin{aligned}
  \text{conjecture}
  &\longrightarrow \text{inequality}
   \longrightarrow \text{counterexample}\\[-1mm]
  &\longrightarrow \text{revised theorem}
   \longrightarrow \text{checked construction}.
  \end{aligned}
\]
That chain is more informative than the number of generated proposals.

It also suggests a concrete research regime for older algorithmic problems.
The favorable cases are not simply those with a short statement.  They have
a white-box objective, inexpensive exact falsification on small instances,
several plausible formulations, limited hidden experimental infrastructure,
and a literature precise enough to define what remains open.  Information
boundaries can then be varied as an experimental condition: a blind attempt
tests rediscovery and structural convergence; a full-context attempt tests
whether inherited knowledge actually changes the bottleneck; an adversarial
condition tests the claims produced by both.

Agents for formal mathematics, evaluator-guided program discovery, and
self-evolving EDA systems emphasize different parts of this regime
~\cite{TaoEtAl2026SolversToResearch,RomeraParedesEtAl2024FunSearch,
DeepMind2025AlphaEvolve,YuRen2026SelfEvolvedABC}.  Prim--Dijkstra required all
three forms of evidence: a proof that could survive reconstruction, a
counterexample that could change the proof, and an implementation whose
frontier could survive a frozen evaluator.  The public artifact exposes a
reproducible subset of those objects, with the boundary stated in
Section~\ref{sec:process}; it is not offered as a complete hidden trace.

The older literature was usable because its definitions, constructions, and
open question could be reconstructed without changing the model.  In this
study that precision let old claims become evaluator cases, proof obligations,
and counterexample searches.  It is a practical precondition for the research
regime above: without stable definitions and checkable objects, additional
agent effort would only multiply incompatible interpretations.

%% file: conclusion.tex
\section{Conclusion}
\label{sec:conclusion}

This revisit settles the terminal-only Manhattan decision problem, gives its
global cost--radius tradeoff a direct \((2,2)\) construction, and carries that
construction into \hprcrst.  The process record adds three concrete
observations: two independently constrained exact attempts reached the same
integrality barrier; one 58-terminal counterexample redirected the
approximation proof; and the final solver retained the theorem-bearing
construction without importing unresolved exact-search tails.  Together they
form an audited case in which language-model agents helped move an older
algorithmic question from reconstruction to proof and implementation, while
proofs, counterexamples, and a frozen evaluator remained the standards of
evidence.

%% file: appendix_npc.tex
\section{Complete Weak NP-Completeness Proof}
\label{app:npc}

This appendix supplies the details suppressed in
Theorem~\ref{thm:npc}.  We reduce from binary \textsc{Partition}.  Its input is
a list of positive integers \(a_1,\ldots,a_n\); the question is whether some
subset sums to \(S/2\), where \(S=\sum_i a_i\).  This problem is NP-complete
under polynomial-time many-one reductions~\cite{GareyJohnson1979}.

\subsection{Membership in NP}

A certificate lists \(n-1\) unordered pairs of terminal indices.  A verifier
rejects invalid indices, loops, and any edge whose endpoints are already
joined in a disjoint-set structure; after the last edge it checks that one
component remains.  It then computes every edge length
\(\lvert x_i-x_j\rvert+\lvert y_i-y_j\rvert\), sums them, and performs one
rooted traversal to compute all root distances.

If \(N\) is the complete binary input length, a coordinate difference, an
edge, and a sum of at most \(n-1\) edges have \(O(N+\log n)=O(N)\) bits.
There are polynomially many exact additions, comparisons, and graph
operations.  The empty certificate for \(n=1\) gives \(L=R=0\).
Consequently the decision problem belongs to NP.

\subsection{The reduction}

Given \(a_1,\ldots,a_n\), put
\[
  S=\sum_i a_i,\qquad h=10S,\qquad t=4S,\qquad X=2hn+t.
\]
Create
\[
  s_i=(2hi,0)\quad(0\leq i\leq n),\qquad
  u_i=(2hi-3a_i,a_i)\quad(1\leq i\leq n),
\]
and \(z=(X,0)\), with root \(s_0\).  The budgets are
\[
  B=X+3S,\qquad D=X+S.                              \tag{11}
\]
Since
\[
  x(u_i)-x(s_{i-1})=2h-3a_i\geq17S>0,\qquad
  x(s_i)-x(u_i)=3a_i>0,
\]
the strict order is
\[
  s_0,u_1,s_1,u_2,s_2,\ldots,u_n,s_n,z.             \tag{12}
\]
All terminals are distinct.  Every coordinate and budget is \(O(nS)\), so
each has \(O(\log n+\log S)\) bits and the mapping has polynomial encoding
length.

\begin{figure*}[t]
  \centering
  \begin{tikzpicture}[
      x=0.84cm,y=0.72cm,
      term/.style={circle,draw=black!65,fill=white,minimum size=4.5mm,
                   inner sep=0pt,font=\footnotesize},
      root/.style={rectangle,draw=black,fill=black,text=white,
                   minimum size=4.5mm,inner sep=0pt,font=\footnotesize},
      axis/.style={line width=.9pt,draw=black!65},
      selected/.style={line width=1.4pt,draw=FDdark},
      unselected/.style={line width=1.4pt,draw=FDred!82!black},
      guide/.style={densely dotted,draw=black!30},
      lab/.style={font=\footnotesize,fill=white,inner sep=1pt}
    ]
    \node[root] (sm) at (0,0) {\(s_{i-1}\)};
    \node[term] (u) at (4.0,1.55) {\(u_i\)};
    \node[term] (s) at (5.5,0) {\(s_i\)};
    \draw[guide] (u)--(4.0,0);
    \draw[axis] (sm)--node[lab,below] {\(2h\)} (s);
    \draw[selected] (sm)--node[lab,sloped,above,text=FDdark!82!black]
      {\(2h-2a_i\)} (u);
    \draw[selected] (u)--node[lab,sloped,above,text=FDdark!82!black]
      {\(4a_i\)} (s);
    \draw[decorate,decoration={brace,amplitude=3pt,mirror},black!45]
      (4.0,-0.45)--node[below=4pt,font=\footnotesize] {\(3a_i\)}
      (5.5,-0.45);
    \draw[decorate,decoration={brace,amplitude=3pt,mirror},black!45]
      (6.0,0)--node[right=4pt,font=\footnotesize] {\(a_i\)}
      (6.0,1.55);
    \node[font=\small\bfseries] at (2.75,2.35)
      {selected block: \(u_is_i\) lies on the root path};

    \begin{scope}[shift={(10.2,0)}]
      \node[root] (sm2) at (0,0) {\(s_{i-1}\)};
      \node[term] (u2) at (4.0,1.55) {\(u_i\)};
      \node[term] (s2) at (5.5,0) {\(s_i\)};
      \draw[guide] (u2)--(4.0,0);
      \draw[unselected] (sm2)--node[lab,below,text=FDred!88!black]
        {\(2h\)} (s2);
      \draw[unselected] (u2)--node[lab,sloped,above,text=FDred!88!black]
        {\(4a_i\)} (s2);
      \node[font=\small\bfseries] at (2.75,2.35)
        {unselected block: \(u_i\) is a leaf};
    \end{scope}
  \end{tikzpicture}
  \caption{The two canonical block choices in the forward reduction.  Blue
  and muted-red paths are combinatorial terminal edges weighted by Manhattan
  distance; the dotted projection is only a coordinate guide.  The reverse
  proof does not assume either drawing and starts from an arbitrary tree in
  the complete terminal metric.}
  \Description{Two schematic block gadgets. In a selected block the path
  enters the upper terminal and then the axis terminal. In an unselected
  block the path goes directly between axis terminals and the upper terminal
  is a leaf.}
  \label{fig:npc-block}
\end{figure*}
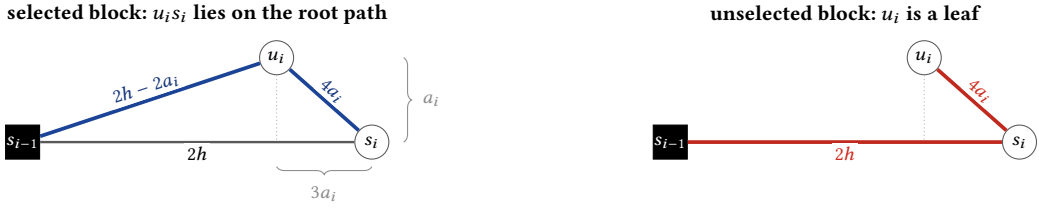

\subsection{Forward implication}

Suppose \(I\subseteq[n]\) has \(A=\sum_{i\in I}a_i=S/2\).  Include
\(u_is_i\) for every \(i\); include \(s_{i-1}u_i\) if \(i\in I\) and
\(s_{i-1}s_i\) otherwise; finally include \(s_nz\).  The construction joins
one new block at a time and is a tree.  Its relevant lengths are
\[
\begin{aligned}
  d(s_{i-1},u_i)&=2h-2a_i,&
  d(s_{i-1},s_i)&=2h,\\
  d(u_i,s_i)&=4a_i,&
  d(s_n,z)&=t.
\end{aligned}
\]
A selected block contributes \(2h+2a_i\) to \(L\), and an unselected block
contributes \(2h+4a_i\).  Hence
\[
  L=X+2A+4(S-A)=X+3S=B.                            \tag{13}
\]
The root--\(z\) path uses \(u_is_i\) exactly for \(i\in I\), and therefore
\[
  d_T(s_0,z)=X+2A=X+S=D.                           \tag{14}
\]
Every vertex on that path is within \(D\).  If \(i\notin I\), the only
off-path vertex in the block is \(u_i\).  With
\(A_i=\sum_{j\in I,j\leq i}a_j\),
\[
  d_T(s_0,u_i)=2hi+2A_i+4a_i
   \leq2hn+S+4S=2hn+t+S=D.                         \tag{15}
\]
Thus the produced tree is feasible.

\subsection{Reverse implication}

Let \(T\) be any feasible tree in the complete Manhattan graph, and let \(P\)
be its unique \(s_0\)-to-\(z\) path.  Give a tree edge multiplicity two if it
lies on \(P\) and one otherwise.  The weighted sum is
\[
  F=L(T)+d_T(s_0,z)\leq B+D=2X+4S.                 \tag{16}
\]

\begin{lemma}[Interblock cut forcing]
\label{lem:interblock}
Exactly one tree edge crosses every prefix cut between two consecutive blocks
in (12), and that edge lies on \(P\).
\end{lemma}

\begin{proof}
For an open vertical gap between consecutive terminals, let \(c_T\) and
\(c_P\) be the numbers of crossing tree and path edges.  Connectivity gives
\(c_T\geq1\), while the path from the leftmost to the rightmost terminal gives
\(c_P\geq1\).  An abstract complete-metric edge contributes its horizontal
length once to every gap between its endpoints, so the horizontal part of
\(F\) is exactly the sum of each gap width times \(c_T+c_P\).  The baseline
over all gaps is \(2X\).

The gap preceding block \(i\) has width at least \(17S>4S\).  A third
crossing would make the horizontal contribution exceed (16).  The final gap
has width \(4S\); a third crossing there would make its horizontal
contribution at least \(2X+4S\), while the vertical contribution is strictly
positive because some \(u_i\) has positive height.  Thus \(c_T+c_P=2\) on
every interblock cut, so \(c_T=c_P=1\).
\end{proof}

\begin{lemma}[Forced topology]
\label{lem:forced-topology}
The tree has exactly one edge between each pair of consecutive blocks, no
other interblock edge, and every internal edge \(u_is_i\).  All interblock
edges lie on \(P\).
\end{lemma}

\begin{proof}
The unique edges of two different nested prefix cuts must be distinct:
deleting one tree edge produces one fixed root-side component, not two
different prefixes.  An edge that skips a block would cross two such cuts and
would have to be their common unique edge.  Therefore the \(n+1\) interblock
edges are distinct and join consecutive blocks.  A tree on \(2n+2\) vertices
has \(2n+1\) edges, leaving exactly \(n\) intrablock edges.  Each two-vertex
block has only the edge \(u_is_i\), so all of them are forced.  Lemma
\ref{lem:interblock} places every interblock edge on \(P\).
\end{proof}

Let \(W\) be the total length of the interblock edges, let
\[
  I=\{i:u_is_i\text{ lies on }P\},\qquad
  A=\sum_{i\in I}a_i.
\]
Lemma~\ref{lem:forced-topology} gives the exact identities
\[
  L(T)=W+4S,\qquad d_T(s_0,z)=W+4A.                \tag{17}
\]

\begin{lemma}[Bridge inequality]
\label{lem:bridge}
\(W+2A\geq X\).
\end{lemma}

\begin{proof}
Orient \(P\) from \(s_0\) to \(z\).  Split \(I\) into \(I_+\), whose internal
edge is traversed \(u_i\to s_i\), and \(I_-\), traversed in reverse; write
\(A_+\) and \(A_-\) for their item sums.  If \(H_b,V_b\) are the total
horizontal and absolute vertical changes on the bridges, then \(W=H_b+V_b\).
Telescoping the signed \(x\)-change along \(P\) yields
\[
  H_b=X-3A_++3A_-.
\]
Telescoping signed \(y\)-change yields bridge signed change \(A_+-A_-\), so
\(V_b\geq A_+-A_-\).  Therefore
\[
\begin{aligned}
  W+2A
  &=H_b+V_b+2(A_++A_-)\\
  &\geq X+4A_-\geq X.
\end{aligned}
\]
Reverse traversal only adds slack.
\end{proof}

Using (17), Lemma~\ref{lem:bridge}, and the cost budget,
\[
  X-2A\leq W\leq B-4S=X-S,
\]
so \(A\geq S/2\).  The path-distance budget gives
\[
  X+2A\leq W+4A=d_T(s_0,z)\leq D=X+S,
\]
so \(A\leq S/2\).  Thus \(A=S/2\), and \(I\) is a valid
\textsc{Partition} witness.  This argument quantifies over every complete
Manhattan edge, including edges with noncanonical consecutive-block
endpoints.

\subsection{Classification}

The construction proves NP-hardness and, with membership in NP,
NP-completeness.  Its coordinates are proportional to \(S\), so the reduction
is numerical and establishes weak NP-hardness.  It neither proves strong
NP-hardness nor rules out a pseudopolynomial algorithm.

%% file: appendix_height.tex
\section{Height Partition and the Factor-Two Boundary}
\label{app:height}

\subsection{Lower bounds and edge-cut facts}

Let \(d\) be any finite metric, let \(r\) be the root, and define
\[
  M=\operatorname{MST}(V,d),\qquad
  \Delta=\max_{v\in V}d(r,v).
\]
Every terminal tree \(T\) satisfies \(M\leq L(T)\) by MST optimality and
\(\Delta\leq R(T)\) by the triangle inequality along each root path.  These
are the two independent lower bounds used by the approximation ratio.

Fix an MST \(T_M\) rooted at \(r\).  For every parent edge
\((u,v)\) of weight \(e_v\), deleting the edge exposes a cut containing \(v\)
but not \(r\).  The complete-metric edge \((r,v)\) crosses this cut, so the
MST cut property gives
\[
  e_v\leq d(r,v)\leq\Delta.                         \tag{18}
\]
The nonnegative net cost of replacing \((u,v)\) by \((r,v)\) is therefore
\[
  a_v=d(r,v)-e_v\leq\Delta.                         \tag{19}
\]

\subsection{Residual invariant}

The mathematical construction permits any real \(H>0\); the integer
implementation evaluates exact integer thresholds.  Process vertices in
postorder.  At the time a vertex \(v\) is processed, its \emph{open residual}
\(C_v\) consists of \(v\) and precisely those descendant residuals whose
joining edges have been retained.  Let
\[
  W_v=\sum_{e\in E(C_v)}d(e),\qquad
  h_v=\max_{x\in C_v}d_{C_v}(v,x).
\]
The algorithm retains \((u,v)\) iff
\[
  e_v+h_v\leq H.                                   \tag{20}
\]
If retained, \(C_v\) is merged into \(C_u\); otherwise it is closed,
\((u,v)\) is removed, and \((r,v)\) is installed.

\begin{lemma}[Residual height]
\label{lem:residual-height}
Every open or closed residual has height at most \(H\).  The final graph is a
terminal spanning tree.
\end{lemma}

\begin{proof}
A leaf residual has height zero.  Inductively, every retained child residual
has height from its parent at most \(e_v+h_v\leq H\), so their union with the
parent also has height at most \(H\).  Cutting an edge closes exactly one
connected residual.  Starting from a tree, each cut increases the number of
components by one and each root shortcut reconnects exactly the newly closed
component; no shortcut connects two vertices already in the same residual.
The final graph is therefore connected with \(n-1\) edges and is a tree.
\end{proof}

\begin{lemma}[Radius]
\label{lem:height-radius}
The returned tree \(A_H\) satisfies \(R(A_H)\leq\Delta+H\).
\end{lemma}

\begin{proof}
The root residual is reached without a shortcut and has height at most \(H\).
Every other residual is reached through its unique root edge \((r,v)\), of
length at most \(\Delta\), followed by a residual path of length at most
\(H\).  Taking the maximum gives the claim.
\end{proof}

\subsection{Disjoint charge}

For a cut child \(v\), use as its resource the deleted edge together with the
edges of its closed residual:
\[
  Q_v=\{(u,v)\}\cup E(C_v),\qquad
  w(Q_v)=e_v+W_v.
\]
The failed test and \(W_v\geq h_v\) imply
\[
  w(Q_v)=e_v+W_v\geq e_v+h_v>H.                    \tag{21}
\]
Combining (19) and (21),
\[
  a_vH\leq\Delta\,w(Q_v).                          \tag{22}
\]

\begin{lemma}[Disjointness]
\label{lem:height-disjoint}
The sets \(Q_v\) over all cut children are pairwise edge-disjoint, and their
total weight is at most \(M\).
\end{lemma}

\begin{proof}
When \(C_v\) is closed, its edges no longer belong to the open residual passed
to any ancestor.  A later cut can charge its own joining edge and the edges
that remain in its open residual, but none already closed below it.  Two cuts
in different branches have disjoint subtree edges.  Hence no MST edge is
charged twice, and all charged edges belong to \(T_M\).
\end{proof}

The output differs from \(T_M\) only at the cut children.  Summing (22) and
using Lemma~\ref{lem:height-disjoint} yields
\[
\begin{aligned}
  \bigl(L(A_H)-M\bigr)H
  &=H\sum_{v\ {\rm cut}}a_v\\
  &\leq\Delta\sum_{v\ {\rm cut}}w(Q_v)
   \leq M\Delta.
\end{aligned}
\tag{23}
\]
Together with Lemma~\ref{lem:height-radius}, this proves
\[
  L(A_H)\leq M\left(1+\frac{\Delta}{H}\right),
  \qquad
  R(A_H)\leq\Delta+H.                              \tag{24}
\]
At \(H=\Delta\), (24) gives \(L\leq2M\) and
\(R\leq2\Delta\).  If \(n=1\), then \(M=\Delta=L=R=0\) and the empty tree is
handled separately.  For \(n>1\) with distinct terminals, \(\Delta>0\).

The output topology changes only when \(H\) crosses one of finitely many
residual heights.  Nevertheless, (24) holds for every real \(H>0\), giving
the continuous certificate envelope plotted in
Figure~\ref{fig:height-tradeoff}.  The released mode samples five integer
thresholds and always includes the exact \(H=\Delta\) member.

\subsection{Why metric and MST arguments alone stop at two}

\begin{theorem}[Arbitrary-metric one-tree lower bound]
\label{thm:metric-lower}
For every \(c<2\), there is a finite positive-integer metric in which every
terminal tree \(A\) satisfies
\[
  \max\{L(A)/M,\ R(A)/\Delta\}>c.
\]
\end{theorem}

\begin{proof}
The case \(c<1\) is immediate, so write \(\varepsilon=2-c\in(0,1]\).
Choose
\[
  q=\max\{2,\lceil4/\varepsilon\rceil-1\},\qquad m=q+1.
\]
Form a generating graph with a unit-edge trunk of length \(2q+1\) from root
\(r\) to a branch vertex, and attach \(m\) disjoint unit-edge arms of length
\(q+1\).  Add an edge of weight \(q\) from \(r\) to every other vertex, and
take the shortest-path metric on all vertices.

The unit trunk-and-arms tree has
\[
  M=2q+1+m(q+1),                                   \tag{25}
\]
because every nonzero metric distance is at least one and it already has
\(n-1\) unit edges.  The root shortcuts give \(\Delta\leq q\), while remote
vertices have distance exactly \(q\), hence \(\Delta=q\).

Take an arbitrary complete-metric terminal tree \(A\).  Expand each edge into
a shortest path in the generating graph, take the union, and then take a root
shortest-path tree \(B\) of that union.  This operation cannot increase either
objective:
\[
  L(B)\leq L(A),\qquad R(B)\leq R(A).               \tag{26}
\]
If \(R(A)>2q\), then \(R(A)/\Delta>2>c\).
Otherwise every path in \(B\) to an arm endpoint has length at most \(2q\).
A path beginning with the unit trunk has length \(3q+2\); one beginning with
a \(q\)-edge to the trunk, branch, or another arm has length at least
\(2q+1\).  Thus reaching each arm endpoint requires a distinct \(q\)-edge
whose nonroot endpoint lies strictly inside that arm.  Tree \(B\) contains at
least \(m\) such edges.  Its other \(n-1-m\) edges have weight at least one,
so
\[
  L(A)\geq L(B)\geq M+m(q-1).                      \tag{27}
\]
Combining the two cases,
\[
  \max\{L(A)/M,R(A)/\Delta\}
  \geq
  1+\frac{m(q-1)}{2q+1+m(q+1)}.                    \tag{28}
\]
The choice of \(q,m\) makes the right-hand side exceed
\(2-\varepsilon=c\).  Indeed \(\varepsilon(q+1)\geq4\), and direct
rearrangement reduces the desired inequality to
\[
  m>
  \frac{(1-\varepsilon)(2q+1)}
       {\varepsilon(q+1)-2},
\]
whose right-hand side is less than \((2q+1)/2<q+1=m\).
The construction is finite; it has
\[
  n=q^2+4q+3<35/\varepsilon^2
\]
vertices.
\end{proof}

This theorem is intentionally qualified.  The metric need not embed in the
Manhattan plane, and the result concerns one tree simultaneously compared
with the MST and star lower-bound witnesses.  Different members of a
portfolio may cover different comparison trees.  Neither a planar-\(L_1\)
factor below two nor an arbitrary-metric portfolio lower bound follows.

%% file: appendix_algorithm.tex
\section{Algorithmic Contract and Exact Subroutines}
\label{app:algorithm}

This appendix fixes the boundary between the mathematical construction, the
bounded empirical search, and the object that is finally emitted.  The
description follows the released implementation rather than an idealized
variant of it.

\subsection{Parent arrays and exact remeasurement}
\label{app:parent-validation}

Let the terminals be indexed by \(V=\{0,\ldots,n-1\}\), with root \(r=0\).
A candidate is a complete parent array
\[
  p[0]=-1,\qquad
  p[v]\in V\setminus\{v\}\quad (v\neq0).
\]
The validator first builds the child lists induced by \(p\).  A traversal
from \(0\) must visit every terminal exactly once; a repeated visit is a
cycle, and an unvisited terminal is disconnected.  Hence a successful array
contains exactly \(n-1\) arcs and represents one rooted terminal tree.

For every successful array, the validator discards all incremental objective
estimates and recomputes
\[
  \ell_v=|x_v-x_{p[v]}|+|y_v-y_{p[v]}|,\qquad
  L(p)=\sum_{v\neq0}\ell_v,
\]
and, during the same rooted traversal,
\[
  q_0=0,\qquad q_v=q_{p[v]}+\ell_v,\qquad
  R(p)=\max_{v\in V}q_v.                                     \tag{29}
\]
Exact parent-array duplicates are removed first.  After remeasurement,
candidates are ordered lexicographically by
\((L,R,p)\); equal metric pairs are coalesced, and a monotone scan in
increasing \(L\) retains a point precisely when its \(R\) is strictly smaller
than every preceding radius.  Thus neither a construction tag nor an
incremental move score can affect feasibility or Pareto dominance.

\paragraph{Arithmetic domain.}
The certified domain consists of distinct terminals whose two coordinate
tokens lie in
\([-2^{63},2^{63}-1]\).  Coordinate differences are widened before
subtraction.  A Manhattan edge is therefore smaller than \(2^{65}\).
Because the implementation uses an \(n<2^{31}\) terminal index, every
\((n-1)\)-edge tree length and every root path is smaller than \(2^{96}\);
signed 128-bit accumulators contain (29) exactly.  Cross-products used for
ratio comparisons, area selection, and certificate inequalities use
arbitrary-precision integers.  No floating-point value participates in
validation, dominance, or witness selection.

The \textsc{certified} mode rejects a coordinate token outside the stated
domain.  For such a token the empirical interfaces return only the valid
terminal star fallback, before entering fixed-width metric arithmetic; that
fallback is not covered by the certified-mode claim.  Duplicate terminals
are rejected in every mode.  The requested portfolio cap is clamped to
\(1\leq k\leq64\).

\subsection{Certificate-preserving selection}
\label{app:certified-selection}

Let \(C\) be the exactly remeasured height-partition tree for
\(H=\Delta\), and let \(\mathcal C\) contain the MST, the star, and all five
height-partition candidates used by the certified mode.  By
Theorem~\ref{thm:height-partition},
\[
  L(C)\leq2M,\qquad R(C)\leq2\Delta.                           \tag{30}
\]
The certified selector is separate from the empirical portfolio selector.
If \(k=1\), it returns the parent array of \(C\) itself.  If \(k>1\), it first
forms the exact nondominated frontier \(\mathcal F\) of \(\mathcal C\) and
chooses
\[
  C^\star\in
  \arg\min_{\substack{T\in\mathcal F\\
                      L(T)\leq L(C),\ R(T)\leq R(C)}}
       \bigl(L(C)-L(T)\bigr)+\bigl(R(C)-R(T)\bigr),            \tag{31}
\]
breaking a remaining tie by the parent array.  Selection begins with
\(C^\star\), adds the two frontier endpoints if space remains, and fills
remaining slots by the exact area rule described below.  It never deletes
\(C^\star\).

\begin{proposition}[Witness preservation]
\label{prop:witness-preservation}
For every supported instance and every \(k\geq1\), the certified output
\(\mathcal P_{\rm cert}\) contains a tree \(T\) satisfying (30).  For \(k=1\),
that tree has exactly the parent array produced by
\textsc{HeightPartition}\((T_M,r,\Delta)\).
\end{proposition}

\begin{proof}
The statement for \(k=1\) is the selector's explicit first branch.  For
\(k>1\), the finite candidate set contains \(C\), so its nondominated
frontier contains either \(C\) or a tree that Pareto-dominates \(C\).
Consequently the feasible set in (31) is nonempty.  The chosen \(C^\star\)
inherits both inequalities in (30) and is inserted before any downsampling;
the certified selector reserves its index throughout truncation.
\end{proof}

There is a second, independent check immediately before output.  It rebuilds
the canonical MST and \(C\), recomputes \(M,\Delta,L(C),R(C)\), remeasures
every returned parent array, and verifies that at least one returned pair
dominates \((L(C),R(C))\).  When \(k=1\), it additionally checks parent-array
identity with \(C\).  This check does not prove
Theorem~\ref{thm:height-partition}; it tests that the executable path has
preserved the theorem's witness.

\subsection{Exact centered partitioning of a fixed MST}
\label{app:centered-dp}

The centered dynamic program is exact for a restricted family that is useful
as an empirical anchor.  Fix a rooted MST \(T_M\), its edge weights \(w\), and
a radius cap \(D\).  Delete any subset \(Q\subseteq E(T_M)\).  The remaining
components must retain their MST edges.  The component containing \(r\) is
centered at \(r\); every other component \(C\) chooses one center
\(c_C\in C\) and receives the direct edge \((r,c_C)\).  Such a partition has
\[
  L=M-\sum_{e\in Q}w(e)
       +\sum_{C\neq C_r}d(r,c_C),                              \tag{32}
\]
and is feasible exactly when
\[
  d(r,c_C)+d_{T_M}(c_C,v)\leq D
  \quad\text{for every }v\in C,                               \tag{33}
\]
with \(c_{C_r}=r\).  Here \(d_{T_M}\) denotes distance in the fixed MST, not
Manhattan distance along a newly optimized tree.

For a vertex \(u\), write \(T_u\) for its rooted subtree and
\(\operatorname{ch}(u)\) for its children.  The open state \(F_u(c)\) is the
minimum net change relative to the retained MST edges inside \(T_u\),
conditioned on \(u\) belonging to an as-yet-open component centered at
\(c\).  The center may lie outside \(T_u\), because that component can
continue through the parent edge.  Set \(F_u(c)=+\infty\) whenever
\[
  d(r,c)+d_{T_M}(c,u)>D.                                      \tag{34}
\]
For a feasible state, define the cost of closing a component at \(u\) by
\[
  G_u=\min_{z\in T_u}\{d(r,z)+F_u(z)\}.                        \tag{35}
\]
The postorder recurrence is
\[
 F_u(c)=
 \sum_{v\in\operatorname{ch}(u)}
 \begin{cases}
   F_v(c),
     & c\in T_v,\\[2mm]
   \min\{F_v(c),\,G_v-w(u,v)\},
     & c\notin T_v,
 \end{cases}                                                  \tag{36}
\]
subject to (34).  A leaf has value \(0\) for every state satisfying (34).
The optimum length in the centered family is
\[
  L_{\rm cent}(D)=M+F_r(r).                                   \tag{37}
\]
The first alternative in (36) retains \((u,v)\).  The second alternative
either retains it, or cuts it, closes the child component using (35), and
subtracts the removed MST edge.  The direct root edge of a closed component
is already included by the term \(d(r,z)\) in (35).

\begin{theorem}[Exactness of the centered DP]
\label{thm:centered-dp}
For fixed \(T_M\) and \(D\), recurrence (34)--(37) returns a minimum-length
tree among all feasible centered partitions of \(T_M\).  Its reconstruction
has radius at most \(D\).  With all pairwise tree distances generated in
quadratic time, the dynamic program uses \(O(n^2)\) time and \(O(n^2)\)
memory.
\end{theorem}

\begin{proof}
Proceed by induction over the rooted subtrees.  Fix a feasible state
\((u,c)\).  If \(c\in T_v\), the unique \(u\)-to-\(c\) path in \(T_M\)
contains \((u,v)\), so that edge must remain and the child has state
\((v,c)\).  If \(c\notin T_v\), either the child remains in the open
component, again giving \(F_v(c)\), or \((u,v)\) is cut.  In the latter case
the entire child subtree is independent of the other children and its top
component must choose some center \(z\in T_v\); (35) gives its optimum direct
attachment and the term \(-w(u,v)\) accounts for the cut.  The child choices
are edge-disjoint, so their net changes add, proving (36).

Conversely, every choice in (36) either keeps or cuts exactly one child edge.
Induction reconstructs a valid centered partition, so the recurrence neither
omits nor introduces a member of the family.  Condition (34), applied at
every vertex of every open component, is precisely (33).  At the root,
\((r,r)\) leaves the root component open without adding a root self-edge,
which proves (37) and the radius statement.

Euler intervals decide \(c\in T_v\) in \(O(1)\) time.  There are \(n^2\)
open states.  Summed over all centers, the child transitions cost
\(n\sum_u|\operatorname{ch}(u)|=O(n^2)\); computing (35), tree distances,
and reconstruction has the same bound.  The table of \(F_u(c)\) dominates
memory.
\end{proof}

The restriction in Theorem~\ref{thm:centered-dp} is substantive: component
interiors must be subtrees of one fixed MST, and each nonroot component uses
one direct root attachment.  The theorem is neither an exact algorithm for
unrestricted RCRST nor an improvement of the universal factor in
Theorem~\ref{thm:height-partition}.

\subsection{Frozen production policies}
\label{app:mode-policies}

All schedules are deterministic functions of \(n\), \(k\), and measured
geometry; neither an instance name nor a stored baseline objective enters
dispatch.  The four public policies are the following.

\begin{description}
  \item[\textsc{certified}.]
    One canonical dense-Prim MST, the star, and height partitions at
    \[
      H\in\{\lfloor\Delta/4\rfloor,\lfloor\Delta/2\rfloor,
             \Delta,2\Delta,4\Delta\},
    \]
    with every positive threshold rounded up to at least \(1\).
    No perturbation, local search, beam, centered DP, or exponential
    enumeration is invoked.

  \item[\textsc{fast}.]
    The scalable schedule uses up to three Prim and three Kruskal tie
    variants, six PD/BRBC/KRY settings, six LAST settings, seven radial
    settings, four centered-MST caps through \(n=128\), cap repair through
    \(n=128\), and bounded repair through \(n=40\).  Above the applicable
    cutoffs, the corresponding families are skipped rather than timed out.

  \item[\textsc{balanced}.]
    For \(n\leq32\), this policy runs the centered primary quality branch with
    96 deterministic perturbation trials but disables the companion
    branch.  For \(n>32\), it uses the same scalable schedule as
    \textsc{fast}.  Exact labeled-tree enumeration via Pr\"ufer codes is
    therefore present in
    \textsc{balanced} for \(n\leq8\), because it belongs to the small primary
    quality engine.

  \item[\textsc{quality}.]
    The primary intensive engine uses 96 trials through \(n=32\), with nine
    LAST and twelve radial settings, critical MST heights, deeper reparent
    and component neighborhoods, and bounded cap beams.  It enumerates the
    exact labeled-tree frontier for \(n\leq8\), applies the centered DP through
    \(n=128\), and retains the richer cap schedule through \(n=256\).
    For \(26\leq n\leq30\), it unions a separately budgeted companion
    branch in which the centered construction family is disabled.
\end{description}

\paragraph{Budget terminology.}
One \emph{perturbation trial} selects a frontier seed, applies one to five
coordinate-seeded reparenting moves, and then performs improving component
exchanges.  A \emph{neighborhood} is the set of trees reachable by one stated
move.  A \(w\times d\) \emph{cap beam} keeps at most \(w\) feasible trees,
ordered by \((L,R,\mathrm{parent})\), at each of \(d\) successive
component-exchange rounds.

The companion branch uses only the ratio between the maximum and median edge
of the canonical MST.  A ratio at least \(8\) receives 96 clustered trials;
a ratio at least \(2\) receives 32 trials with the deeper companion
neighborhood; all other cases receive 32 trials with the shallower
neighborhood.  The primary and companion defaults are, respectively, 48 and
at most 48 reparent seeds, two and at most two \emph{additional} deep
reparent layers after the first two neighborhood expansions, 32 and at most
32 component seeds, two and at most two component rounds, and an
\(8\times2\) and at most \(8\times2\) cap beam.  The shallower companion
reduces these configurable budgets to 32 reparent seeds, one additional
layer, 24 component seeds, one round, and a \(6\times1\) beam.

For clarity, the retained size boundaries are collected here:
\[
\begin{array}{ll}
  \text{mechanism} & \text{production boundary}\\
  \text{complete labeled-tree frontier} & n\leq8\\
  \text{primary intensive search} & n\leq32\\
  \text{centered-disabled companion branch} & 26\leq n\leq30\\
  \text{centered fixed-MST DP} & n\leq128\\
  \text{quality cap schedule} & n\leq256\\
  \text{additional MST/BRBC families} & n\leq512
\end{array}                                                    \tag{38}
\]
The centered DP is evaluated at four caps
\[
  D=\Delta+H,\qquad
  H\in\{\lfloor0.7\Delta\rfloor,\lfloor0.75\Delta\rfloor,
          \lfloor0.8\Delta\rfloor,\Delta\},
\]
again replacing a zero positive-instance threshold by \(1\).

When \(k<64\), the scalable engine imports its cheaper endpoint schedule
before filtering, and the quality engine imports the scalable schedule.
Quality also imports the scalable schedule throughout the transition band
\(33\leq n\leq40\), even at \(k=64\).  These unions protect lower-cost
anchors when a small output cap changes which nondominated seeds survive;
they do not make the empirical policies approximation algorithms.

\subsection{General portfolio truncation}
\label{app:portfolio-selection}

Suppose the exact nondominated frontier is
\[
  T_0,T_1,\ldots,T_{m-1},
  \quad L(T_0)<\cdots<L(T_{m-1}),\quad
        R(T_0)>\cdots>R(T_{m-1}).
\]
If \(m\leq k\), it is returned unchanged.  Otherwise the empirical selector
reserves the two endpoints.  For each candidate marked essential by an
anchor or lower-policy import, it also reserves a frontier point that
dominates the marked candidate and minimizes exact additive slack.  If these
reservations themselves exceed \(k\), endpoints are retained first and the
remaining reserved indices are admitted in frontier order until the cap is
reached.  Thus unconditional witness preservation is a property of the
separate certified selector, not of an arbitrarily small empirical
portfolio.

For every consecutive pair \(T_a,T_b\) already selected, each unselected
\(T_i\), \(a<i<b\), receives the exact triangle score
\[
  A(a,i,b)=
  \left|
    (L_b-L_a)(R_i-R_a)-(L_i-L_a)(R_b-R_a)
  \right|.                                                     \tag{39}
\]
The selector repeatedly inserts the point with largest \(A\), breaking ties
by frontier index, until \(k\) points have been selected.  Equation (39) is a
deterministic shape-preserving downsampling rule; it is not a hypervolume
optimizer and carries no approximation claim.

\subsection{Claim and formalization boundaries}
\label{app:algorithm-boundaries}

The guarantees used in the paper have three distinct statuses.

\begin{itemize}
  \item The weak NP-completeness result and the height-partition inequalities
        are mathematical theorems proved in this paper.  The executable
        certificate checks are independent arithmetic audits of a produced
        witness, not machine-checked proofs of those theorems.
  \item The centered dynamic program is exact only for the fixed-MST family
        defined by (32)--(33).  Complete labeled-tree enumeration via
        Pr\"ufer codes is exact for unrestricted terminal trees only at the
        enforced cutoff \(n\leq8\).
  \item Only \textsc{certified}, on distinct signed-64-bit-coordinate
        terminals, has the universal \((2,2)\) guarantee.  Claims for
        \textsc{fast}, \textsc{balanced}, and \textsc{quality} are empirical
        claims about the released benchmark records.
\end{itemize}

The released Lean snapshot contains completed, \texttt{sorry}-free modules
for the finite tree model and elementary lower bounds, Manhattan parity, and
the gcd length lattice.  It does not contain completed formal proofs of the
NP-completeness reduction, the height-partition theorem, the metric lower
bound, or Theorem~\ref{thm:centered-dp}.  Accordingly, none of those results
is presented as formally verified.

%% file: appendix_experiments.tex
\section{Experimental Protocol and Complete Frozen Results}
\label{app:experiments}

This appendix fixes the meanings of every reported comparison and records the
complete aggregate results in the release bundle.  The machine-readable JSONL
files, rather than the rounded values below, are authoritative.

\subsection{Independent evaluator}

An instance consists of distinct integer-coordinate terminals
\(V=(v_0,\ldots,v_{n-1})\), with \(v_0\) as the root.  A solver output is a
parent array \(p\), where \(p_0=-1\).  Before using an output, the evaluator
checks the parent range, rejects self loops and cycles, verifies that every
chain reaches \(v_0\), and then recomputes
\[
  L(p)=\sum_{i=1}^{n-1}\|v_i-v_{p_i}\|_1,\qquad
  R(p)=\max_{i} \sum_{(u,w)\in p[v_0,v_i]}\|u-w\|_1
  \tag{D.1}
\]
using Python integers.  Values printed by the solver are never trusted.
Duplicate parent arrays and duplicate \((L,R)\) pairs are removed.  Sorting
the remaining pairs by increasing \(L\), a pair is retained exactly when its
\(R\) is strictly smaller than every earlier radius; this produces the exact
nondominated frontier \(\mathcal F(P)\).

For two portfolios \(P\) and \(Q\), define
\[
  \begin{aligned}
  P\succeq Q
  &\quad\Longleftrightarrow\quad
  \forall q\in\mathcal F(Q)\ \exists p\in\mathcal F(P):\\
  &\hspace{18mm}L(p)\leq L(q)\ \land\ R(p)\leq R(q).
  \end{aligned}
  \tag{D.2}
\]
The four reported relations are
\emph{dominance} (\(P\succeq Q\), \(Q\not\succeq P\)),
\emph{tie} (\(P\succeq Q\), \(Q\succeq P\)),
\emph{mixed} (neither), and
\emph{loss} (\(Q\succeq P\), \(P\not\succeq Q\)).
Thus ``tie'' means equality after Pareto closure, not equality of parent
arrays or generator tags.

Hypervolume is also fixed per instance.  If the largest coordinates in the
frozen reference frontier are \(L_{\max}\) and \(R_{\max}\), the common corner
is
\[
  \begin{aligned}
  C_L&=L_{\max}+\max\{1,\lfloor L_{\max}/20\rfloor\},\\
  C_R&=R_{\max}+\max\{1,\lfloor R_{\max}/20\rfloor\}.
  \end{aligned}
\]
For frontier points \((L_i,R_i)\) in increasing-\(L\) order, put
\(R_0=C_R\) and compute
\[
  \operatorname{HV}_C(P)
    =\sum_{i:\,L_i<C_L,\ R_i<R_{i-1}}
       (C_L-L_i)(R_{i-1}-R_i).
\]
The per-instance score is
\[
  \Delta\operatorname{HV}(P)
  =\frac{\operatorname{HV}_C(P)-\operatorname{HV}_C(P_{\rm ref})}
         {C_LC_R}.
\]
All area calculations before the final division are exact integers.  We
report both the sum over 28 instances and its mean; these two aggregations
must not be compared as though they were the same quantity.

\paragraph{Two different common factors.}
The empirical common factor of a whole portfolio is
\begin{equation}
 c_{\rm port}(I,P)=
   \min_{T\in P}\max\left\{\frac{L(T)}{M_I},
                           \frac{R(T)}{\Delta_I}\right\}.
 \label{eq:app-cport}
\end{equation}
By contrast, the certificate factor is attached to the particular
\(H=\Delta_I\) height-partition tree \(C_I\):
\begin{equation}
 c_{\rm cert}(I)=
   \max\left\{\frac{L(C_I)}{M_I},
              \frac{R(C_I)}{\Delta_I}\right\}\leq2.
 \label{eq:app-ccert}
\end{equation}
Equation~\eqref{eq:app-cport} may select the MST, star, or another height
threshold, and is therefore normally smaller than
\eqref{eq:app-ccert}.  If \(C_I\) is Pareto dominated, the
returned portfolio may retain an exactly remeasured dominator instead; the
proof obligation is still checked against the reconstructed \(C_I\).
No benchmark in this appendix is degenerate (\(M_I\Delta_I>0\)).

\subsection{Datasets and frozen reference}

The 28-instance development set contains 4--32 terminals: one smoke case, one
SALT toy case, four extracted Superblue nets, uniform and two-cluster random
families, two grids, and two near-line cases.  It was visible during
development and is used for regression and baseline comparison, not as a
blind generalization set.

The separate 50-instance set contains 6--100 terminals from uniform,
three-cluster, near-line, grid, comb, and radial families.  The scaling set
contains five geometries---uniform, four-cluster, near-line, comb, and
radial---at each
\(n\in\{64,128,256,512,1024,2048\}\), for 30 instances total.
Finally, exact validation uses three six-terminal instances
(\(6^4=1296\) labeled trees each) and one eight-terminal instance
(\(8^6=262144\) labeled trees).

The 28-case published-method reference is the nondominated union of three
sources:

\begin{enumerate}
  \item terminal-only star, RMST, PD, BRBC, and KRY outputs from the pinned
        SALT snapshot;
  \item the official PD-Rev/PD-II implementation at the 19 values
        \(\alpha_1=\alpha_2=k/20\), \(k=1,\ldots,19\), flipping distance two,
        with HVW/DAS disabled;
  \item the official 2023 MSPD/MSS implementation at the 11 settings
        \(\alpha=j/10\), \(j=0,\ldots,10\), and the 16 pairs
        \((\kappa,\lambda)\in\{1,\ldots,4\}^2\), exported immediately before
        HVW/DAS Steinerization.
\end{enumerate}

Legacy inputs are translated by a common per-instance offset before being
passed to the author programs; Manhattan distances are unchanged.  Every
author output is converted to a terminal parent array and independently
remeasured.  After Pareto closure, the published union contains 139 stored
points: 56 classic, 51 PD-II, and 32 MSPD/MSS points.  Environment~04 is not
part of this union; its outputs are used only in the separately labeled
internal-predecessor stress test below.

\subsection{Raw-data provenance}

The frozen exporter ran all modes with a portfolio limit of 64 and a
120-second per-invocation timeout.  The quality and held-out sets retain one
run per case and mode; scaling retains three runs for each of 30 cases and
four modes.  There is no discarded warm-up.  Wall time starts immediately
before process creation and ends after process reap and output flush;
per-process peak RSS is obtained from \texttt{wait4} when available.  Thread
environment variables are fixed to one, and the release build uses portable
\texttt{-O3 -DNDEBUG} settings without architecture-specific flags.

The result environment records source commit \texttt{bb4e95c0e08}; the
complete hash is retained in the machine-readable metadata.  The raw bundle contains
112 development records, 200 held-out records, 360 scaling records, 78
certificate reconstructions, four exact-oracle records, and 168 ablation
records.  Each invocation record includes the command, complete input,
stdout, stderr, exit status, timeout status, wall time, independently measured
tree metrics, and frontier.  Environment, hypervolume-corner, compiler, and
verification metadata are stored separately.  The CSV and Markdown summaries
are generated only from these raw records.

\subsection{Frozen quality results}

\begin{table*}[t]
  \caption{All four modes against the 28-case published-method union.}
  \label{tab:app-frozen-quality}
  \centering
  \small
  \begin{tabular}{lrrrrrr}
    \toprule
    Mode & Dominance & Tie & Mixed & Loss &
    \(\sum_I\Delta\mathrm{HV}\) &
    \(\operatorname{mean}_I\Delta\mathrm{HV}\)\\
    \midrule
    certified & 0  & 1 & 1 & 26 &
      \(-0.7478\) & \(-0.026708\)\\
    fast          & 12 & 6 & 9 & 1 &
      \(+0.0654\) & \(+0.002337\)\\
    balanced      & 23 & 5 & 0 & 0 &
      \(+0.0863\) & \(+0.003081\)\\
    quality       & 23 & 5 & 0 & 0 &
      \(+0.0863\) & \(+0.003081\)\\
    \bottomrule
  \end{tabular}
\end{table*}

Table~\ref{tab:app-frozen-quality} measures the entire frontier.  In
particular, the published union covers the certified portfolio on 26
instances, is mixed with it on one, and ties it on one; the factor-two theorem
does not imply frontier coverage.  Balanced and quality cover the union on
every instance and add at least one strict Pareto improvement on 23.

\subsection{Internal-predecessor stress test}
\label{app:internal-predecessor-stress}

As a separate diagnostic, we supplemented the published-method union with
the frozen output of Environment~04, an earlier solver produced within this
project.  This deliberately stronger stress-test union is neither an external
baseline nor a state-of-the-art claim, and it is not used for the main
comparisons above.

\begin{table*}[t]
  \caption{Diagnostic comparison with the published union augmented by the
  internal Environment-04 predecessor.}
  \label{tab:app-internal-predecessor}
  \centering
  \small
  \begin{tabular}{lrrrrr}
    \toprule
    Mode & Dominance & Tie & Mixed & Loss &
    \(\sum_I\Delta\mathrm{HV}\)\\
    \midrule
    certified & 0  & 1  & 0 & 27 & \(-0.6038\)\\
    fast          & 0  & 11 & 5 & 12 & \(-0.0127\)\\
    balanced      & 13 & 15 & 0 & 0  & \(+0.0050\)\\
    quality       & 13 & 15 & 0 & 0  & \(+0.0050\)\\
    \bottomrule
  \end{tabular}
\end{table*}

\begin{table*}[t]
  \caption{Common-factor statistics.  ``Certified portfolio'' uses
  Eq.~\eqref{eq:app-cport}; ``\(H=\Delta\) certificate'' uses
  Eq.~\eqref{eq:app-ccert}.}
  \label{tab:app-common-factors}
  \centering
  \scriptsize
  \setlength{\tabcolsep}{3.3pt}
  \begin{tabular}{llrrrr}
    \toprule
    Set & Object & Mean & Median & Min. & Max.\\
    \midrule
    development & certified portfolio
      & 1.21 & 1.24 & 1.00 & 1.53\\
    development & published portfolio
      & 1.06 & 1.05 & 1.00 & 1.18\\
    development & \(H=\Delta\) certificate
      & 1.27 & 1.25 & 1.00 & 1.94\\
    \midrule
    held-out & certified portfolio
      & 1.25 & 1.22 & 1.00 & 1.97\\
    held-out & \(H=\Delta\) certificate
      & 1.31 & 1.25 & 1.00 & 1.97\\
    \bottomrule
  \end{tabular}
\end{table*}

On the development set, the certified portfolio contains one to six
nondominated trees (median three); on the held-out set it contains one to
seven (median four).  The worst development portfolio factor, \(1.53\),
occurs on \texttt{uniform\_n30\_s2}.  Independent certificate reconstruction
finds no failure among all 78 development and held-out cases.  The returned
portfolio contains the original \(H=\Delta\) tree in 66 cases (23 development
and 43 held-out) and an exact Pareto dominator in 12 (5 and 7,
respectively).

\subsection{Held-out ordering and absolute time}

\begin{table*}[t]
  \caption{Coverage transitions on all 50 held-out instances.  A transition
  \(A\!\to\!B\) compares candidate \(B\) with reference \(A\).}
  \label{tab:app-heldout-coverage}
  \centering
  \small
  \begin{tabular}{lrrrr}
    \toprule
    Transition & Dominance & Tie & Mixed & Loss\\
    \midrule
    certified \(\to\) fast & 44 & 3  & 3 & 0\\
    fast \(\to\) balanced      & 21 & 29 & 0 & 0\\
    balanced \(\to\) quality   & 9  & 41 & 0 & 0\\
    \bottomrule
  \end{tabular}
\end{table*}

\begin{table*}[t]
  \caption{Held-out common factor and one-process wall time.  Times include
  startup and serialization; one run is retained per case.}
  \label{tab:app-heldout-time}
  \centering
  \small
  \begin{tabular}{lrrrrr}
    \toprule
    Mode & Mean \(c_{\rm port}\) & Median \(c_{\rm port}\) &
    Maximum \(c_{\rm port}\) & Median time (ms) & Maximum time (ms)\\
    \midrule
    certified & 1.25 & 1.22 & 1.97 & 6.9 & 13.7\\
    fast      & 1.11 & 1.08 & 1.44 & 13.0 & 262.0\\
    balanced  & 1.11 & 1.08 & 1.44 & 62.2 & 19030.5\\
    quality   & 1.11 & 1.08 & 1.44 & 63.3 & 19140.5\\
    \bottomrule
  \end{tabular}
\end{table*}

The large maxima for balanced and quality arise from their bounded intensive
small-instance tiers; they are retained rather than hidden by a cutoff.
Table~\ref{tab:app-heldout-coverage} is a portfolio statement, whereas the
first three numeric columns of Table~\ref{tab:app-heldout-time} summarize one
balanced point per instance.

\subsection{Scaling}

\begin{table*}[t]
  \caption{Median end-to-end wall time in milliseconds on the scaling set.
  Each entry is the median of five geometries and three retained runs.}
  \label{tab:app-scaling-all}
  \centering
  \small
  \begin{tabular}{rrrrr}
    \toprule
    \(n\) & certified & fast & balanced & quality\\
    \midrule
      64 &  6.1 &   8.7 &   8.7 &   9.8\\
     128 &  6.3 &  17.0 &  16.9 &  19.9\\
     256 &  6.1 &  15.3 &  15.6 &  26.2\\
     512 &  9.4 &  18.1 &  18.5 &  22.7\\
    1024 & 13.3 &  42.2 &  41.9 &  54.2\\
    2048 & 27.1 & 140.8 & 142.2 & 182.1\\
    \bottomrule
  \end{tabular}
\end{table*}

\begin{table*}[t]
  \caption{Complete \(n=2048\) endpoint.  Peak RSS is the maximum observed
  per-process resident set; tree count is the median.}
  \label{tab:app-scaling-endpoint}
  \centering
  \scriptsize
  \setlength{\tabcolsep}{3pt}
  \begin{tabular}{lrrrr}
    \toprule
    Mode & Median ms & Max. ms & Trees & Peak MiB\\
    \midrule
    certified & 27.1 &  32.5 & 5 & 2.48\\
    fast      & 140.8 & 182.8 & 8 & 3.66\\
    balanced  & 142.2 & 220.4 & 8 & 3.75\\
    quality   & 182.1 & 302.3 & 8 & 4.05\\
    \bottomrule
  \end{tabular}
\end{table*}

All 360 scaling invocations are valid and completed within the timeout.
These are candidate-only measurements; the scaling data do not contain
baseline runs.

\subsection{Exactness and component ablation}

Complete labeled-tree enumeration via Pr\"ufer codes gives exact frontiers
for all four small instances.  Fast, balanced, and quality tie the oracle on all four;
certified loses on all four, as expected for its deliberately sparse
construction family.  The production test suite has 26 unit and property
tests, and the frozen verification record reports zero quality-gate
regressions, zero certificate failures, and successful builds of both recent
baseline adaptations.

\begin{table*}[t]
  \caption{Complete compact ablation against the 28-case published-method
  union.  Runtime columns
  are from the same ablation run and are included to expose its variable host
  load; quality comparisons do not use them.}
  \label{tab:app-ablation-complete}
  \centering
  \scriptsize
  \setlength{\tabcolsep}{3.2pt}
  \begin{tabular}{lrrrrrrr}
    \toprule
    Variant & Dom. & Tie & Mix. & Loss &
    \(\sum\Delta\mathrm{HV}\) & Median ms & Max. ms\\
    \midrule
    production
      & 23 & 5 & 0 & 0 & \(+0.0863\) & 848.8 & 37991.6\\
    construction only
      & 1 & 9 & 8 & 10 & \(-0.0589\) & 16.5 & 602.3\\
    scalable tier only
      & 12 & 7 & 9 & 0 & \(+0.0665\) & 72.0 & 1008.9\\
    no deep search
      & 20 & 6 & 2 & 0 & \(+0.0840\) & 670.2 & 5981.2\\
    no component exchange
      & 23 & 5 & 0 & 0 & \(+0.0857\) & 706.4 & 30436.9\\
    reduced starts
      & 23 & 5 & 0 & 0 & \(+0.0862\) & 785.7 & 38378.8\\
    \bottomrule
  \end{tabular}
\end{table*}

The variants are compiled from the same source.  ``Construction only''
removes heuristic search; ``scalable tier only'' removes the intensive small-case
tier; the next two rows remove the named search stages; and ``reduced
starts'' lowers deterministic search trials and seed caps.  Because
constructions also seed later neighborhoods, these rows are interventions,
not additive attributions.

\subsection{Released-artifact boundary}

The public repository releases the solver source and tests; the three
benchmark JSONL files; the independent evaluator and result exporters; the
frozen reference parent arrays; redistributable source, licenses, and
provenance for the PD-II and MSPD/MSS adaptations; selected 01--05 prompts and
task contexts; the completed, \texttt{sorry}-free Lean Basic, Parity, and Gcd
modules; and the frozen raw outputs and generated summaries described above.

It does \emph{not} redistribute the pinned SALT source: only SALT's validated
terminal-tree contribution to the frozen union is present.  It also does not
claim to release an exhaustive agent transcript or hidden model state.
Incomplete Lean developments for NP-completeness, height partition, and the
approximation lower bound are omitted, so the main theorems are not claimed
as formally verified.  These boundaries are part of the reproducibility
statement, not exceptions to it.

%% file: appendix_record.tex
\section{Released Research Record and Formalization Boundary}
\label{app:record}

This appendix states exactly which parts of the model-assisted research record
are public.  It distinguishes a \emph{task specification}, which records what
an investigation was allowed to see and asked to do, from an execution
transcript, which would record everything that happened while carrying out that
task.  The repository releases the former for five conditions, but not the
latter.

\subsection{Model setting and five task specifications}

All five investigations used OpenAI GPT-5.6 Sol in Codex with ultra reasoning
effort.  They differed in their assigned objective, information boundary, and
minimum persistence requirement, not in the selected model or reasoning
setting.

Each directory under \texttt{prompts/} contains exactly three files:
\texttt{PROMPT.md}, \texttt{AGENTS.md}, and \texttt{PROBLEM.md}.  The first two
record the intended objective and access policy; the third gives the common
problem definition.  At a high level, the five conditions were:

\begin{description}
  \item[\texttt{01}: hardness, closed world.]
  The condition supplied the problem definition but no references, algorithms,
  benchmarks, or Internet access.  It asked for an adversarially audited
  NP-completeness proof, including the complete-Manhattan-graph shortcut audit,
  and required contemporaneous claim, approach, and counterexample ledgers.

  \item[\texttt{02}: exact tractability, blind.]
  This condition fixed the premise that a polynomial-time constructive route
  exists, prohibited hardness arguments and early abandonment, and supplied no
  papers or existing algorithms.  It allowed a public benchmark and a
  command-line evaluator only as falsification aids, forbade inspecting or
  modifying the evaluator, and imposed an eight-active-hour persistence
  protocol.

  \item[\texttt{03}: exact tractability, locally informed.]
  The mathematical premise and persistence requirement matched
  \texttt{02}, but this condition could consult a curated local
  \texttt{resources/} collection and supplied baseline code.  Network access
  and sibling experiments remained forbidden, and every transferred claim or
  code path required a provenance and model-match audit.  The contrast with
  \texttt{02} therefore isolates the effect of prior literature and code
  without changing the requested complexity conclusion.

  \item[\texttt{04}: approximation, literature open.]
  This condition took the hardness result from \texttt{01} as a fixed premise
  and redirected effort toward a strong bicriteria approximation and an
  improved empirical frontier.  Focused Internet search for public
  mathematical papers and implementations was permitted and source-logged;
  the frozen evaluator and supplied baselines remained immutable.  The prompt
  required at least eight active hours and continued post-pass work on the
  theorem, algorithm, runtime, and held-out behavior.

  \item[\texttt{05}: full-context engineering synthesis.]
  The final condition could inspect the accumulated evidence from
  \texttt{01}--\texttt{04}, the local reference and baseline collections, and
  relevant public online material.  It was asked to build and audit the
  strongest practical multi-mode solver, with exact output validation,
  controlled comparisons, ablations, and scaling measurements.  Its
  four-active-hour budget was a minimum rather than a deadline, and the prompt
  deliberately left the final algorithmic architecture open.
\end{description}

These descriptions concern the \emph{specified} information boundaries.  The
released prompt directories do not include the original \texttt{worklog/} or
\texttt{submission/} trees, evaluator-private state, downloaded reference
collections, integrity manifests, or filesystem images mentioned by those
specifications.  In particular, the presence of detailed ledger requirements
inside a prompt must not be read as a claim that the corresponding complete
ledger has been released.  The public packages document the experimental
conditions; they are neither full conversation histories nor containerized
replays of the five runs.

\subsection{What the public repository contains}

The released repository contains the final deterministic solver and its build
file; production, property, exact-small, and certificate tests; the checked-in
quality, held-out, and scaling instances; an independent tree/frontier checker
and result exporters; and the source and provenance material needed for the
released PD-II and 2023 MSPD/MSS terminal-output adaptations.  It also contains
the frozen nondominated baseline union, algorithm and evaluation notes, the five
task specifications above, the completed Lean snapshot described below, and
machine-readable raw results together with tables generated from those raw
records.

It does \emph{not} release complete model conversations, complete
condition-by-condition research histories, every intermediate claim ledger or
counterexample ledger, every exploratory program, private reference copies,
unfinished Lean developments, or complete execution-environment images.
Likewise, the frozen recent-baseline artifact has no timing data collected
under the solver's process-level timing protocol; it supports objective-space
comparisons, not a same-machine runtime ratio.  These omissions are why the
artifact is described as a reproducibility subset of the research record
rather than as ``everything produced'' during the project.

\subsection{Lean snapshot and theorem boundary}

The directory \texttt{formal/} contains three completed, \texttt{sorry}-free
modules:

\begin{itemize}
  \item \texttt{Rcrst/Basic.lean} defines finite metrics and rooted trees in the
  parent/depth model, total length, root radius, and direct radius.  It proves
  the pointwise direct-distance lower bound, \(\directradius\leq
  \treeradius(T)\), and that the star attains
  \(\treeradius=\directradius\); it also defines the MST optimality predicate.
  \item \texttt{Rcrst/Parity.lean} constructs the integer Manhattan metric and
  proves that every root-to-terminal tree-path length has the same parity as
  the corresponding direct Manhattan distance.
  \item \texttt{Rcrst/Gcd.lean} proves the length-lattice lemma: if \(g\)
  divides every pairwise distance, then \(g\) divides every tree length, with
  the associated exact budget-rounding lemma.
\end{itemize}

The weak NP-completeness reduction in Appendix~\ref{app:npc}, the
height-partition theorem in Appendix~\ref{app:height}, and the associated
lower-bound construction are \emph{not} completed in Lean and are not included
as Lean theorem files.  The main results of this paper are consequently
mathematical proofs checked in the ordinary manner, not claims of end-to-end
formal verification.  The Lean snapshot is also intentionally outside the
default solver build and experimental verification target.

\subsection{Frozen measurements and verification record}

The \texttt{results/raw/} bundle stores per-instance machine-readable outputs
for the 28-case comparison, 50 held-out cases, scaling runs through
\(n=2048\), certificate reconstruction, exact-small checks, and component
ablations.  Invocation records retain the command, input, standard output and
error, exit status, and wall time; scaling records additionally retain
per-process peak resident memory when the host exposes it.  The published CSV
and Markdown summaries are regenerated only from these raw records.

The frozen verification record separately captures the portable optimized
build, the 26-test suite, the public quality gate, and compilation of the two
redistributable recent-baseline source trees.  Its environment record includes
the source revision, compiler and flags, operating-system and processor
metadata, single-thread settings, timeout, timing boundary, and repeat policy.
Lean is explicitly outside that verification target, consistently with the
formalization boundary above.  Together these files make the reported
measurements inspectable and rerunnable without implying that an unreleased
research transcript has been reconstructed.